\documentclass[a4paper, english]{lipics-v2019}
\usepackage[T1]{fontenc}
\usepackage[utf8]{inputenc}
\usepackage{csquotes}
\usepackage{microtype}
\usepackage{mathtools}
\usepackage{amsmath}
\usepackage{amssymb}
\usepackage{nccmath}
\usepackage{booktabs}
\usepackage{multirow}
\usepackage{algorithm}

\usepackage{tikz}
\usetikzlibrary{matrix}
\usetikzlibrary{graphs}
\usetikzlibrary{shapes}
\usetikzlibrary{patterns}
\usetikzlibrary{decorations,arrows}
\usetikzlibrary{decorations.pathmorphing}
\usepgflibrary{decorations.pathreplacing} 
\usetikzlibrary{calc}
\usetikzlibrary{arrows.meta}

\newcommand{\drawJob}[5][$ $]{%
\draw[fill = white] (#2,#3) rectangle node[midway]{#1} (#4,#5)}

\newcommand{\drawMachine}[5][$ $]{%
	\draw[pattern = north west lines] (#4,#5+#3) -- (#4,#5) -- (#4+#2,#5) -- (#4+#2,#5+#3);
	\node[below] at (#4+0.5*#2,0){#1};} 

\usepackage[margin,noinline,draft]{fixme}

\newcommand{\nfold}{$N$-fold}

\usepackage{amsmath}
\usepackage{wrapfig}

\makeatletter
\let\save@mathaccent\mathaccent
\newcommand*\if@single[3]{%
  \setbox0\hbox{${\mathaccent"0362{#1}}^H$}%
  \setbox2\hbox{${\mathaccent"0362{\kern0pt#1}}^H$}%
  \ifdim\ht0=\ht2 #3\else #2\fi
  }
\newcommand*\rel@kern[1]{\kern#1\dimexpr\macc@kerna}
\newcommand*\widebar[1]{\@ifnextchar^{{\wide@bar{#1}{0}}}{\wide@bar{#1}{1}}}
\newcommand*\wide@bar[2]{\if@single{#1}{\wide@bar@{#1}{#2}{1}}{\wide@bar@{#1}{#2}{2}}}
\newcommand*\wide@bar@[3]{%
  \begingroup
  \def\mathaccent##1##2{%
    \let\mathaccent\save@mathaccent
    \if#32 \let\macc@nucleus\first@char \fi
    \setbox\z@\hbox{$\macc@style{\macc@nucleus}_{}$}%
    \setbox\tw@\hbox{$\macc@style{\macc@nucleus}{}_{}$}%
    \dimen@\wd\tw@
    \advance\dimen@-\wd\z@
    \divide\dimen@ 3
    \@tempdima\wd\tw@
    \advance\@tempdima-\scriptspace
    \divide\@tempdima 10
    \advance\dimen@-\@tempdima
    \ifdim\dimen@>\z@ \dimen@0pt\fi
    \rel@kern{0.6}\kern-\dimen@
    \if#31
      \overline{\rel@kern{-0.6}\kern\dimen@\macc@nucleus\rel@kern{0.4}\kern\dimen@}%
      \advance\dimen@0.4\dimexpr\macc@kerna
      \let\final@kern#2%
      \ifdim\dimen@<\z@ \let\final@kern1\fi
      \if\final@kern1 \kern-\dimen@\fi
    \else
      \overline{\rel@kern{-0.6}\kern\dimen@#1}%
    \fi
  }%
  \macc@depth\@ne
  \let\math@bgroup\@empty \let\math@egroup\macc@set@skewchar
  \mathsurround\z@ \frozen@everymath{\mathgroup\macc@group\relax}%
  \macc@set@skewchar\relax
  \let\mathaccentV\macc@nested@a
  \if#31
    \macc@nested@a\relax111{#1}%
  \else
    \def\gobble@till@marker##1\endmarker{}%
    \futurelet\first@char\gobble@till@marker#1\endmarker
    \ifcat\noexpand\first@char A\else
      \def\first@char{}%
    \fi
    \macc@nested@a\relax111{\first@char}%
  \fi
  \endgroup
}
\makeatother

\usepackage{amsthm}
\newtheorem*{lemma*}{Lemma}

\usepackage{xspace}
\newcommand{\ie}{i.\,e.\xspace}

\newcommand{\Wlogeneral}{W.\,l.\,o.\,g.\xspace}

\newcommand\Class[1]{%
    \mathchoice%
    {\text{\normalfont\fontsize{9pt}{10pt}\selectfont$\mathrm{#1}$}}%
    {\text{\normalfont\fontsize{9pt}{10pt}\selectfont$\mathrm{#1}$}}%
    {\text{\normalfont$\mathrm{#1}$}}%
    {\text{\normalfont$\mathrm{#1}$}}%
  }

\newcommand{\Lang}[1]{%
  \ifmmode{%
    \text{\normalfont\textsc{#1}}%
  }%
  \else
  {\normalfont\textsc{#1}}%
  \fi}

\usepackage[smaller]{acronym}
\newacro{EPTAS}{efficient polynomial time approximation scheme}
\newacro{FPTAS}{fully polynomial time approximation scheme}
\newacro{ILP}{integer linear program}
\newacro{LP}{linear program}
\newacro{ETH}{exponential time hypothesis}

\newcommand{\ints}{\mathbb{Z}}

\DeclarePairedDelimiter{\norm}{\lVert}{\rVert}

\DeclarePairedDelimiter\floor{\lfloor}{\rfloor}
\DeclarePairedDelimiter\ceil{\lceil}{\rceil}
\DeclarePairedDelimiter\set{\lbrace}{\rbrace}
\DeclarePairedDelimiterX\sett[2]{\lbrace}{\rbrace}{ #1 \,\delimsize| \,\mathopen{} #2 }

\newcommand{\cO}{\mathcal{O}}

\usepackage{cryptocode}

\hideLIPIcs
\nolinenumbers

\authorrunning{K. Jansen, A. Lassota, M. Maack}

\acknowledgements{We thank Malin Rau for helpful discussions on the problem.}

\author{Klaus Jansen}{Department of Computer Science, Kiel University,  Kiel, Germany}{kj@informatik.uni-kiel.de}{}{Supported by supported by DFG Project "Strukturaussagen und deren Anwendung in Scheduling- und Packungsprobleme", JA 612/20-1}
\author{Alexandra Lassota}{Department of Computer Science, Kiel University,  Kiel, Germany}{ala@informatik.uni-kiel.de}{}{Supported by DFG Project "Strukturaussagen und deren Anwendung in Scheduling- und Packungsprobleme", JA 612/20-1}
\author{Marten Maack}{Department of Computer Science, Kiel University,  Kiel, Germany}{mmaa@informatik.uni-kiel.de}{}{}
\title{Approximation Algorithms for Scheduling with Class Constraints}
\begin{document}
\maketitle

\begin{abstract}
Assigning jobs onto identical machines with the objective to minimize the maximal load is one of the most basic problems in combinatorial optimization. Motivated by product planing and data placement, we study a natural extension called Class Constrainted Scheduling (CCS). In this problem, each job additionally belongs to a class and each machine can only schedule jobs from at most $c$ different classes. Even though this problem is closely related to the Class Constraint Bin Packing, the Class Constraint Knapsack and the Cardinality Constraint variants, CCS lacks results regarding approximation algorithms, even though it is also well-known to be NP-hard. We fill this gap by analyzing the problem considering three different ways to feasibly allot the jobs: The splittable case, where we are allowed to split and allot the jobs arbitrarily; the preemptive case, where jobs can be split but pieces belonging to the same job are not allowed to be scheduled in parallel; and finally the non-preemptive case, where no splitting is allowed at all. 
For each case we introduce the first PTAS where neither $c$ nor the number of all classes have to be a constant. In order to achieve this goal, we give new insights about the structure of optimal solutions. This allows us to preprocess the instance appropriately and by additionally grouping variables to set up a configuration Integer Linear Program (ILP) with \nfold{} structure. This \nfold{} structure allows us to solve the ILP efficiently yielding the desired running times. 
Further we developed the first simple approximation algorithms with a constant approximation ratio running in strongly polynomial time. The splittable and the preemptive case admit algorithms with ratio~$2$ and a running time of $\cO(n^2 \log(n))$. The algorithm for the non-preemptive case has a ratio of $7/3$ and a running time of $\cO(n^2 \log^2(n))$. All results even hold if the number of machines cannot be bounded by a polynomial in $n$. 

\end{abstract}

\section{Introduction}
One of the most basic questions in combinatorial optimization is to assign jobs onto identical, parallel machines such that the makespan is minimized. In this paper, we study a natural extension where each job additionally admits a class and each machine can only schedule jobs from a limited number of different classes. This problem is called Class Constrained Scheduling. Its additional property arises commonly when considering product planing or data placement. For example, consider operations which need access to a database. Due to logistical limitations or time restrictions, the databases have to be stored locally. However, disk space is limited, hence we cannot afford to store all databases on each machine. Thus, we have to guarantee that for each job scheduled on a machine we are able to store all required information, which will be a true subset of all databases.

Formally, we are given $n \in \mathbb{N}$ jobs $J$, each job $j \in J$ with processing time $p_j \in \mathbb{N}$ and class $c_j \in \{1, \dots, C\}$. These jobs have to be allotted onto $m$ identical machines~$M$, each with a limitation of $c$ class slots. A class slot can contain any number of jobs from one arbitrary class. In other words, the allotted jobs on a machine can be from at most $c$ different classes. Clearly we may assume $c \leq C$ and $c \leq n$. Otherwise we either allow more classes to be scheduled on one machine than we have overall or we allow more classes to be scheduled on one machine than we have jobs. Hence in both cases, all jobs can go together on a single machine yielding the classical scheduling problem. Furthermore, we assume $C \leq n$ as we can discard classes without any belonging jobs.

In this paper we study three different variants of a feasible job placement:
\begin{itemize}
\item \emph{Splittable case:} In the splittable case, we are allowed to cut the jobs into arbitrary small pieces and place them anywhere as long as we do not schedule two jobs at the same time on the same machine and only assign jobs from at most $c$ different classes onto each machine.  Formally, we define a function $\pi \colon J \rightarrow \mathbb{Z}_{\geq 0}$ which maps every job to the number of pieces it is split into. Further we define a function $\lambda_j \colon [\pi(j)] \rightarrow (0,1]$ such that $\sum\nolimits_{k \in \pi(j)} \lambda_j(k) = 1$ which states the fraction of the overall processing time of job $j$ for each part. Define $\mathcal{J}= \{(j, p) \,|\, j \in J, p \in [\pi(j)]\}$ as the set of job parts. An assignment $\mu \colon \mathcal{J} \rightarrow M$ matches job pieces onto machines. Finally, we define a schedule $\sigma = (\pi, \lambda, \mu)$. The makespan of a schedule is defined by the maximum sum of processing times on a machine, \ie max$_{i \in M}\{\sum\nolimits_{(j,p) \in \mu^{-1}(i)} \lambda(p)p_j \}$.
\item \emph{Preemptive case:} This case resembles the splittable case, but in addition the fractions of the same job are not allowed to be scheduled in parallel, \ie at the same time on different machines. Thus we cannot solely assign job pieces onto machines as before, in addition we also need to state the starting points of the fractions. Let $\xi \colon \mathcal{J} \rightarrow \mathbb{Q}_{\geq 0}$ be a function defining the starting times of a job piece. A schedule is defined as $\sigma = (\pi, \lambda, \xi, \mu)$ and it has to hold that for each two job parts $(j, p), (j', p') \in \mathcal{J}$ with $\mu(j, p) = \mu(j', p')$ or $j = j'$ that $\xi(j,p) + \lambda_j(p)p_j \leq \xi(j', p')$ or $\xi(j',p') + \lambda_{j'}(p')p_{j'} \leq \xi(j, p)$. In other words, job pieces belonging to the same job are not allowed to be scheduled in parallel.
\item \emph{Non-preemptive case:} This case does not allow any splitting of the jobs. Hence the definition of a schedule is simpler as we can directly map jobs onto machines. In detail, a schedule $\sigma \colon J \rightarrow M$ assigns each job $j \in J$ onto a machine $ \sigma(j) = i \in M$. The makespan~$\mu(\sigma)$ of a schedule $\sigma$ is defined by the maximum sum of processing times a machine has to schedule, \ie max$_{i \in M}\{\sum\nolimits_{j \in \sigma^{-1}(i)} p_j \}$.
\end{itemize}

The input of each case can be described as an instance $I = [p_1, \dots, p_n, c_1, \dots, c_n, m, c]$. The output will be described by the schedule $\sigma$. In the splittable case as well as in the preemptive case, the output length may be huge as we can produce arbitrarily many job pieces. However, we managed for all algorithms in this paper to bound the output length by a polynomial in $n$ either by carefully decoding the output or by bounding the number of produced job pieces appropriately. Overall, we aim to find a schedule $\sigma$ such that the makespan is minimized while assigning jobs of at most $c$ different classes onto a machine. We denote an optimal schedule with minimum makespan by $\textsc{opt}(I)$.

Finding $\textsc{opt}(I)$ is well-known to be $\Class{NP}$-hard for all cases \cite{epstein2010class, shachnai2001polynomial}. Thus we are satisfied to produce a solution close to the makespan of an optimal schedule. However while doing so we have to deal with a trade-off between running time and approximation ratio. This means, we have to spend more time to compute a solution closer to the makespan of an optimal schedule. This can also be seen in our results. First, we present simple algorithms which produce a solution $\sigma$ efficiently with a guaranteed quality of $\mu(\sigma) \leq 2 \cdot \mu(\textsc{opt}(I))$ for the splittable case and the preemptive case. For the non-preemptive case we adapt the algorithms to obtain a ratio of $7/3$. To the best of our knowledge, these are the first algorithms to produce solutions with a guaranteed constant approximation ratio for this problem. Second, we also investigate \emph{approximation schemes}, that upon an input~$\epsilon$ compute a $(1 + \epsilon)$-approximation, \ie it is guaranteed that we find a schedule~$\sigma_A$ with $\mu(\sigma_A) \leq (1+ \epsilon) \cdot \mu(\textsc{opt}(I))$ arbitrarily close the makespan of an optimal solution. In this work we managed to establish the first polynomial time approximation scheme (PTAS) for each case. A PTAS is an approximation algorithm~$A$ receiving $I$ and an approximation value~$\epsilon \in (0, 1]$ as input and computes a schedule while the running time is bounded by $\cO(|I|^{f(1/\epsilon)})$. Here, $f$ is an arbitrary function and |I| is the encoding length. For the investigated problems the encoding length $|I|$ of an instance $|I|$ is given by $|I| = \cO(\sum\nolimits_{i=1}^n \lceil$log$(p_i) \rceil+ \sum\nolimits_{i=1}^n \lceil$log$(c_i) \rceil + n + \lceil \log(m) \rceil$, where log$(x)$ denotes the logarithm with basis two of a number $x$. Note that $c$ and $C$ are dominated by $n$ and thus do not appear in the landau-notation. The running time of the PTAS for the splittable case is given by $n^{\cO(1/\delta^4 \log(1/\delta))} m\log(m) \log(p_{\max})$ if the number of machines can be bounded by a polynomial in the number $n$ of jobs. Otherwise a careful adaption yields a running time of $n^{\cO(1/\delta^4 \log(1/\delta))} \log(m) \log(p_{\max})$. The preemptive case admits an PTAS with running time of $n^{2^{\cO(1/\delta^2)}} \log(m) \log(p_{\max})$ and the non-preemptive case of $n^{\cO(1/\delta^8 \log(1/\delta))} \log(m) \log(p_{\max})$. To obtain the PTASs we make use of special ILPs called \nfold{}s. These are Integer Linear Programs where the constraint matrix is of a specific block structure. In Detail, non-zero entries only appear block-wise in the first few rows and the diagonal underneath. This form allows us to compute the solution efficiently. However, it is not straight-forward to formulate our problems in this form. First, we need some preprocessing steps to bound the number of small job, which are the ones harder to handle. Furthermore, we need some novel insights into the structure of optimal solutions. Only then we can set up a configuration ILP with \nfold{} structure by hierarchically grouping the variables. Finally we use the structural properties to transform the solution of the \nfold{} ILP into a feasible schedule.  

 \subsection*{Related Work and Our Results}
Scheduling and packing problems are broadly studied regarding various aspects. For a survey we recommend~\cite{chen1998review, christensen2016multidimensional, kellerer2003knapsack}. 

First note that the makespan minimization on parallel machines, that is, the non-preemptive problem studied in this paper but without class constraints, is strongly NP-hard and well-known to admit a PTAS \cite{hochbaum1987using}.
For CCS approximation schemes are known for two special cases.
On the one hand, Shachnai and Tamir~\cite{shachnai2001polynomial} presented a PTAS for the case that number of classes $C$ is a constant, and, on the other, Chen et al.~\cite{JansenLZ16} designed a PTAS for the case that each class contains exactly one job.
The latter result even works if the class constraints are machine depended, that is, for each machine $i$ a capacity bound $c_i$ is given.

In the remainder of this section, we discuss the known results for the class constrained versions of Bin Packing (CCBP) and Multiple Knapsack (CCKP).
These problems are closely related to CCS and indeed NP-hard~\cite{epstein2010class, shachnai2001two}, and thus studying approximation schemes is a natural approach to obtain solutions efficiently. 

Golubchik et al. present in~\cite{golubchik2009approximation} a PTAS for CCKP for the case that all knapsacks are identical. If $C$ is a constant, there is also a PTAS for the general case by Shachnai and Tamir introduced in~\cite{shachnai2001polynomial}. Furthermore, they present in the same work an FPTAS --- a PTAS with running time also polynomial in $1/\epsilon$ --- for 0-1 CCKP when all items are distinct~\cite{shachnai2001polynomial}. However, the problem becomes APX-hard when each item admits a set of classes~\cite{shachnai2001polynomial}, thus there is no PTAS unless P = NP.

Xavier et al. prove in~\cite{xavier2008class} that CCBP admits an APTAS if $C$ is a constant, \ie a PTAS which additionally allows an additive error. Furthermore, there is no APTAS when $c$ is not constant~\cite{epstein2010class}. In the same paper Epstein et al.~\cite{epstein2010class} present an AFPTAS for the case of a constant $C$, improving upon the APTAS of Xavier et al. An AFPTAS is an APTAS where the running time is not only polynomial in the number of items but also in $1/\epsilon$.
A special sub-case for CCBP is the Cardinality Constrained version, where $C = n$ and thus the bound on the number of classes on in a bin corresponds to the bound on the number of items. This problem admits an APTAS \cite{caprara2003approximation} and an AFPTAS \cite{epstein2010afptas}.

We revisit the Class Constrained Scheduling problem. In this paper, we consider the case where $C$ and $c$ are parameters and not constants as assumed in the previous works. To tackle this problem, we make use of special ILPs, called \nfold{}s. These are Integer Linear Programs with a specific block structure regarding the constrained matrix. Meaning, small blocks only appear in the first few rows and the diagonal underneath. First, it was shown by De Loera et al.~\cite{de2008n} that these ILPs can be solved in polynomial time for constant parameters. This algorithm was greatly improved by a line of interesting results \cite{ DBLP:journals/corr/abs-1904.01361, klein, hemmecke2013n, DBLP:conf/icalp/JansenLR19, koutecky}. The currently best algorithms to solve this ILPs run in near-linear time, whereas in \cite{DBLP:conf/icalp/JansenLR19} the exponential dependency on the parameters $r,s$ is lower and in \cite{DBLP:journals/corr/abs-1904.01361} the polylogarithmic terms are smaller. The idea of using \nfold{}s to construct approximation schemes first appeared in \cite{DBLP:conf/innovations/JansenKMR19}. The authors obtain an \nfold{} by carefully setting up the configuration ILP for the Scheduling problem with setup times, \ie each job consists of (a machine dependent) processing time and a (machine dependent) setup time. The objective is to minimize the makespan. We use a similar approach for setting up the \nfold{}s. However due to the divergent problem setting, the preprocessing and postprocessing steps differ hugely. Futher it was necessary to prove some structural results regarding optimal solutions. But having them at hand, we managed to obtain the following results:
 
 \begin{itemize}
 \item We present the first constant approximation algorithm for the splittable case of the CCS problem with running time $\cO(n^2 \log(n))$ and a quality of $2$, even if the number of machines cannot be bounded by a polynomial in $n$.
 \item We introduce the first constant approximation algorithm with a quality of 2 for the preemptive case of the CCS problem by extending the first algorithm to handle jobs which may be scheduled in parallel. The running time of the extended algorithm is $\cO(n^2 \log(n))$.
 \item We state the first constant approximation algorithm with a quality of $7/3$ for the non-preemptive version of the CCS problem by extending the first algorithm. The algorithm handles the large classes more carefully as they cannot be split. The overall running time of this algorithm is bounded by $\cO(n^2 \log^2(n))$.
 \item Further we introduce the first PTAS for the splittable case by modeling the problem as an \nfold{} ILP. We obtain a running time of $n^{\cO(1/\delta^4 \log(1/\delta))} m\log(m) \log(p_{\max})$ for the case that the number of machines can be bounded by a polynomial in $n$. When setting up the ILP, we need some new insights into the structure of an optimal solution and a careful grouping of the variables. Again, we also managed to handle the case, that the number of machines cannot be bounded by a polynomial in $n$ yielding a running time of $n^{\cO(1/\delta^4 \log(1/\delta))} \log(m) \log(p_{\max})$.  Note that there is no exponential dependency on $c$ or $C$ in both running times.
 \item Adapting our techniques and proving new structural results regarding optimal solutions, we also established the first PTAS for the preemptive case. The overall running time of the PTAS is given by $n^{2^{\cO(1/\delta^2)}} \log(m) \log(p_{\max})$, again being polynomial in $C$ and $c$. 
 \item Lastly, we also present the first PTAS for the non-preemptive case where neither $c$ nor $C$ have to be constant. It admits a running time of $n^{\cO(1/\delta^8 \log(1/\delta))} \log(m) \log(p_{\max})$.
 \end{itemize}

 \subsection*{Structure of the Document}
In Section~\ref{CCS:ILP} we familiarize with \nfold{}s, which are ILPs admitting a specific block structured constrained matrix. These \nfold{}s will be used later on to construct the PTASs. Section~\ref{CCS:ConstantApproximationAlgorithms} presents the simple approximation algorithms for efficiently computing a solution with constant quality. Improving upon the approximation guarantee, we show the construction of these problems as \nfold{}s in Section~\ref{CCS:PTAS} yielding the desired PTASs. Finally, Section~\ref{CSS:OpenQuestions} states some interesting open questions.
 
 \section{\nfold{} Integer Linear Programming} \label{CCS:ILP}
In the following section we introduce a special type of Integer Linear Programs called \nfold{}s formally and state the main results regarding them necessary for understanding the polynomial time approximation schemes. These \nfold{}s consist of a specific structure, where non-zero entries only appear in the first rows and in small blocks along the diagonal underneath. Let $N,r,s,t \in \mathbb{Z}^+$. Furthermore let $A_1, \dots, A_N$ be integer matrices of size $r \times t$ and the integer matrices $B_1, \dots, B_N$ be of size $s \times t$. The constraint matrix $\mathcal A$ is then of following form:
\begin{center}
$
\mathcal A =
\begin{pmatrix}
A_1	& A_2	& \dots	& A_N      \\
B_1	& 0 	& \dots  	& 0 	  \\
0	&B_2	&\dots 	& 0	\\
\vdots	& \vdots 	& \ddots & \vdots \\
0 	& 0 & \dots 	 & B_N
\end{pmatrix}
$  
\end{center}

The constraint matrix $\mathcal A$ has dimension $(r+N\cdot s) \times N\cdot t$.  We will divide $\mathcal A$ into \emph{blocks} of size $(r+s) \times t$. Similarly, the variables of a solution $x$ are partitioned into \emph{bricks} of length~$t$. Hence, each brick $x^{(i)}$ corresponds to the columns of one submatrix $A_i$ and therefore also one submatrix $B_i$.

Given $w, l, u \in \mathbb{Z}^{N\cdot t}$ and $b \in \mathbb{Z}^{r+N\cdot s}$, the corresponding Integer Linear Programming problem is defined by
\begin{align*}
 \textrm{ min}\,\{wx|\, \mathcal Ax = b, l \leq x \leq u,\, x \in \mathbb{Z}^{N\cdot t} \} .
\end{align*}

Currently the best known algorithm for solving this problem with respect to the dependency on the parameters $r, s$ and a near-linear dependency on $nt$ is stated in the next theorem.

\begin{theorem}[\cite{DBLP:conf/icalp/JansenLR19}] \label{nfold}
The \nfold{} ILP is solvable in time $(rs\Delta)^{\cO(r^2s + s^2)} L \cdot Nt \log^{\cO(1)}(Nt)$ when all lower and upper bounds are finite. Here, $\Delta$ is the largest absolute number in $\mathcal A$ and $L$ is the encoding length of the largest number in the complete input.
\end{theorem} 

In the following we will call constraints appearing in $A_1, \dots A_N$ globally uniform constraints. Constraints corresponding to $B_1, \dots, B_N$ will be called locally uniform, since the entries in $b$ can differ for each block. Obviously $r$ matches the number of globally uniform constraints, analogously $N \cdot s$ corresponds the number of the locally uniform ones.

 
\section{Constant Approximation Algorithms} \label{CCS:ConstantApproximationAlgorithms}
As a warm-up, we start with the approximation algorithms admitting a constant quality. These algorithms trade off a worse guaranteed quality of a solution in return for simplicity and an efficient running time. Each of the following algorithms share the same framework with just minor adaptions for the distinct cases: Guess the optimal makespan $T$; group the jobs belonging to classes which admit a accumulated processing time greater than $T$ into as few sub-groups as an optimal schedule has to use; distribute all sub-groups via round robin, a cyclic approach, such that a minimum of class slots is used. This guarantees a feasible solution. Further we can show that doing so we get a $2$-approximation algorithm for the splittable case and for the preemptive case. As for the non-preemptive version this yields a $7/3$-approximation, since we are not allowed to cut the jobs and thus introduce a larger error while grouping them. The splittable and the preemptive case admit algorithms with a running time of $\cO(n^2 \log(n))$. As for the non-preemptive case we get a running time of  $\cO(n^2 \log^2(n))$ as the grouping takes more time in this case.

\subsection*{Splittable Case}
First, we discuss the splittable case, where we are allowed to cut and distribute the jobs arbitrarily between the machines as long as the class restriction is not violated and jobs do not overlap on a machine. Let the overall running time $P_u$ of a class $u \in [C]$ be the accumulated processing time of all jobs belonging to this class, \ie $P_u = \sum\nolimits_{\{j \,|\, c_j = u\}} p_j$. The algorithm for solving the splittable version is displayed in Algorithm \ref{alg:AlgoSplit}.

\begin{algorithm}
\begin{center}
 \fbox{\pseudocode[head={Input: Processing times $p_{1}, \ldots,p_{n}$, \newline corresponding classes $c_1, \ldots, c_n$, \newline number of machines $m$, \newline number of class slots $c$ on each machine}, width = 9.3cm]{\\
 	\text{Calculate $P_u$ for each class $u$.}\\
	\text{Calculate the lower bound $LB = \sum_{j=1}^{n} p_j/m$ and the upper bound \textit{UB}$ = c\cdot \max_u\{P_u\}$.}\\
	\text{Do a binary search between $LB$ and \textit{UB}, where $T$ is the current guess of the makespan:}\\
		\t \text{For each class with $P_u > T$ do: Divide the class into $C_u = \lceil P_u / T \rceil$ new, unique classes,}\\
		\t[2] \text{each with $P_{u'} \leq T$.}\\
	\t \text{Delete divided classes.}\\
	\t \text{If the number of present classes is greater than $ c\cdot m$: Increase the current guess $T$.}\\
	\t \text{Otherwise lower the guess $T$.}\\
	\text{Compute $P_{u'}$ for the new classes.}\\
	\text{Sort the present classes in non-ascending order regarding $P_{u'}$.}\\
	\text{Allot the classes in non-ascending order regarding $P_{u'}$ onto the machines via round robin. (*)}\\
	\text{Reassign the original classes to the jobs.}\\
	\text{Output schedule.}
   }
 }
  \end{center}
 \caption{Algorithm for solving the splittable version of the Class Constrained Scheduling problem. This algorithm also serves as a framework for the preemptive and non-preemptive setting of this problem. The label (*) will serve as an orientation point for the changes when handling the other settings.} \label{alg:AlgoSplit}
 \end{algorithm}
 
The algorithm searches for the optimal makespan via a binary search. In each iteration the classes with $P_u > T$ are divided into $\lceil P_u / T \rceil$ new, unique classes with accumulated processing time smaller or equal to $T$. Therefore all classes considered for the remaining algorithm have $P_{u'} \leq T$. Splitting the classes with $P_u > T$ is indeed easy, as we simply cut each class into pieces of size $T$ until the last fraction with load smaller or equal to $T$ remains. Next, the algorithm checks if the number of remaining and new classes exceeds $c \cdot m$. If so, we discard the guess. Otherwise the current guess is lowered until we find the lowest feasible guess $T$ satisfying this condition. Then all original classes with $P_u \leq T$ and all new classes will be distributed among the machines via \emph{round robin}. This is a procedure where the classes are placed in non-ascending order regarding $P_{u'}$, such that the first class is assigned onto the first machine, the second one onto the second machine and so on until all machines admit a class. This is then repeated until all classes are assigned. An example is given in Figure \ref{fig:RoundRobin}. Note that this procedure is independent of the guess $T$. However, we prove in Theorem \ref{t:ConstSplit} that this algorithm indeed computes a feasible 2-approximation for the splittable version of the Class-Constrainted Scheduling problem.
 
But first we have to look more closely at the binary search. Here we cannot simply test the mid of the current interval and then cut it in half as the optimal makespan can be fractional. Thus we would need infinite many steps or have to be satisfied with a finite precision. Gladly, we can circumvent this obstacle by formulating our binary search a little smarter. Recall that we divide large classes with $P_u > T$ into some new classes, each with $P_{u'} \leq T$. Regarding the algorithm -- whose correctness is proven below -- the only obstacle preventing us from computing a feasible schedule is the amount of different classes after splitting the large ones. Thus the only interesting guesses on the makespan are those values, where by guessing below the number of new classes increases. In the following we will call them \emph{borders}. Now we only have to search along these borders to find the optimal makespan. We do so in the following way: For a class $u$ compute the first $m$ borders via $P_u / k$ for $k \in \{1, \dots m\}$; search in a binary search manner along these borders and save the smallest, feasible guess; do so for each class; output the smallest feasible border of all classes. The following lemma will prove the correctness this advanced binary search: 

\begin{lemma} \label{t:BinarySearch} 
The binary search explained above computes a value smaller or equal to the optimal makespan in time $\cO(C \log(m))$.
\end{lemma}
\begin{proof}
First of all, let us discuss the upper and lower bound. The lower bound  $LB = \sum_{j=1}^{n} p_j/m$ corresponds to an equal distribution of processing times onto the machines, which clearly bounds the optimal makespan from beneath. Further, as a total of $c$ classes fit onto a single machine, $c$ times the maximal accumulated processing time, \ie \emph{UB}$ = c*\max_u\{P_u\}$, states a correct upper bound. Thus it is correct to search for the optimal makespan in between these bounds.

Further it is sufficient to consider just the borders as solely the number of produced classes prevent a guess from being successful. As per definition this number does not change between two neighboring borders, we thus aim to find the lowest border where we do not create more than $c \cdot m$ classes.

Computing $P_u / k$ for $k \in \{1, \dots m\}$ for a class $u$ indeed corresponds to the borders regarding that class, \ie where the class $u$ is split into the same number of sub-classes. This can be seen as for each guess $T \in [P_u/(k+1), P_u/k)$ we get $k$ pieces of size $T$ and one piece of size smaller or equal $T$ yielding $k+1$ classes for all $k \in \{1, \dots, m-1\}$. Remark that we further have an interval $[P_u, UB]$ where class $u$ is not cut at all. Moreover, we do not have to consider larger values for $k$ as it would include that we already produced $m$ new classes with load~$T$ and thus the total area of our current guess would be exceeded. This contradicts the lower bound.

Proceeding with a binary search along the borders of one class results in the correct lowest value. A binary search is applicable as below each border the number of sub-classes increases monotonically. Setting the lowest guess on the smallest value among all classes gives the correct value as when computing the number of sub-classes for a certain class and border all other classes are also considered. So this corresponds to the smallest number resulting in at most $c \cdot m$ classes overall.

Regarding the running time we have to compute at most $\log(m)$ bounds which are visited while doing the binary search for each class. Then we just take the smallest number. Overall this takes time $\cO(C \log(m))$. 
\end{proof}

Next we prove a general lemma regarding round robin which comes in handy when proving the approximation ratio of our algorithm. 

\begin{lemma} \label{t:RoundRobin} 
Let $\sigma$ be a schedule produced by round robin when packing $n$ jobs with processing times $p_1, \dots, p_n$ onto $m$ maschines. Then $\mu(\sigma) \leq \sum_{j=1}^n p_j/m + \max\sett{p_j}{j\in[n]}$.   
 \end{lemma}
 \begin{proof}
We prove this lemma by leading the opposite assumption to a contradiction: 
Consider the loads $L_i$ for each machine $i \in \{1, 2, \dots, m\}$. 
Assume there is one machine $i^*$ with $L_{i^*} > \sum_{j=1}^n p_j/m + \max\sett{p_j}{j\in[n]}$. We will show that under this assumption each machine has a load greater than $\sum_{j=1}^n p_j/m$ and hence $\sum_{j=1}^n p_j = \sum_{i=1}^{m}L_i > m\sum_{j=1}^n p_j/m = \sum_{j=1}^n p_j$.
If we remove the largest job $j$ placed on $i^*$, then $i^*$ still has a load of at least $T$.
\Wlogeneral we may assume that $j$ was placed in the first iteration of round robin.
Let $L'_i$ be the load any machine $i$ receives after $j$ was placed.
We have $L'_{i+1} \geq \dots \geq L'_{m} \geq L'_1 \geq \dots L'_{i} > T$ since the jobs are assigned ordered decreasingly by their processing times. This is graphically displayed in Figure \ref{fig:RoundRobin}.
Hence, we have $L_i \geq L'_i > T$ thereby completing the proof. 
\end{proof}

\begin{center}
\begin{figure}
	\begin{tikzpicture}
	\pgfmathsetmacro{\w}{1}
	\pgfmathsetmacro{\h}{0.5}
	\pgfmathsetmacro{\mH}{11}
		
	\drawMachine[$m_1$]{1*\w}{\mH*\h}{0}{0};
	\pgfmathsetmacro{\hh}{0}
	\foreach \x/\z in {
		5/1,
		3/5,
		1/9
	}{
		\pgfmathsetmacro{\hhh}{\hh+\x}
		\drawJob[\z]{0*\w}{\hh *\h}{\w}{\hhh*\h};	
		\global\let\hh=\hhh
	}

	\pgfmathsetmacro{\ww}{2}
	\drawMachine[$m_2$]{1*\w}{\mH*\h}{\ww*\w}{0};
	\pgfmathsetmacro{\hh}{0}
	\foreach \x/\z in {
		5/2,
		2/6,
		1/10
	}{
		\pgfmathsetmacro{\hhh}{\hh+\x}
		\drawJob[\z]{\ww*\w}{\hh *\h}{\ww*\w +\w}{\hhh*\h};	
		\global\let\hh=\hhh
	}
	
	\pgfmathsetmacro{\ww}{4}
	\drawMachine[$m_3$]{1*\w}{\mH*\h}{\ww*\w}{0};
	\pgfmathsetmacro{\hh}{0}
	\foreach \x/\z in {
		4/3,
		2/7
	}{
		\pgfmathsetmacro{\hhh}{\hh+\x}
		\drawJob[\z]{\ww*\w}{\hh *\h}{\ww*\w +\w}{\hhh*\h};	
		\global\let\hh=\hhh
	}

	\pgfmathsetmacro{\ww}{6}
	\drawMachine[$m_4$]{1*\w}{\mH*\h}{\ww*\w}{0};
	\pgfmathsetmacro{\hh}{0}
	\foreach \x/\z in {
		3/4,
		1/8
	}{
		\pgfmathsetmacro{\hhh}{\hh+\x}
		\drawJob[\z]{\ww*\w}{\hh *\h}{\ww*\w +\w}{\hhh*\h};	
		\global\let\hh=\hhh
	}

	\draw[dashed] (-1*\w,5*\h) node[left]{$T$} -- (8*\w,5*\h); 
	\draw[dashed] (-1*\w,10*\h) node[left]{$2T$} -- (8*\w,10*\h); 
	\draw[decorate,decoration={brace,amplitude=4pt}] (0*\w,5*\h) -- (0*\w,9*\h) node[midway,left,xshift=-0pt,align=center]{$P \geq $};
	\draw[decorate,decoration={brace,amplitude=4pt}] (7*\w,4*\h) -- (7*\w,0*\h) node[midway,right,xshift=4pt,align=center]{$ =: P$};
	\end{tikzpicture}
\caption{The figure shows an example of the round robin scheduling approach. The classes are sorted and numbered regarding their total processing times. Further it is highlighted that the load $L'_i$ of a machine $i$, here $i = 1$ when removing the largest job $j$, here $j = 1$, satisfies $P = L'_{i+1} \geq \dots \geq L'_{m} \geq L'_1 \geq \dots L'_{i} \leq P$.} \label{fig:RoundRobin}
\end{figure}
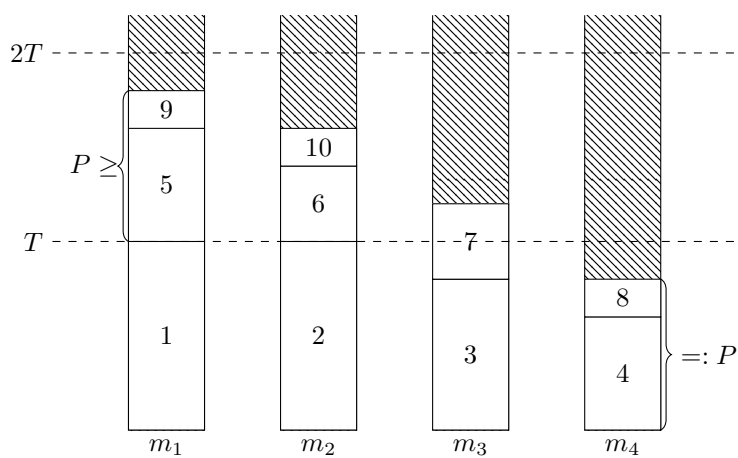
\end{center}

\begin{theorem} \label{t:ConstSplit}
The approximation algorithm above produces a solution $\sigma$ with $\mu(\sigma) \leq 2 \cdot \mu(\textsc{opt}(I))$ for the splittable case in time $\cO(n^2 \log(n))$.
 \end{theorem}
 \begin{proof}
By using the advanced binary search we proved in Lemma \ref{t:BinarySearch} that we check the relevant guesses.  In each iteration with guess $T$ we divide each class with $P_u > T$ into  $C_u = \lceil P_u / T \rceil$ new ones simply by cutting it into pieces of size $T$. Obviously $C_u$ new classes are sufficient. Each new class will have $P_{u'} \leq T$. Manipulating the instance in such a manner is legitimate. Firstly, if there is a feasible schedule with makespan $T$, then each class with $P_u > T$ will use at least $C_u = \lceil P_u / T \rceil$ class slots as we need at least this many machines to distribute its total load. Second, our algorithm allots each of the new classes and the original small ones as a whole, hence we use the absolute minimum of class slots any schedule with makespan at most $T$ could occupy. Thus we either obtain a feasible schedule as $c\cdot m$ class slots are sufficient or there is no schedule with makespan $T$ and we can for a certainty discard this guess.

If the number of class slots are sufficient, we allot the new classes and the original ones with $P_u \leq T$ using round robin. Note that this procedure is independent of the guess $T$. In the last step we then reassign the original classes via a simple mapping. This produces no error, as we have enough class slots and reinserting the original ones can solely decrease the number of used class slots. This occurs, when two new classes map to the same original class. Further as we use exactly the same space as reserved for the new class, the load is not increased.

We argued that we obtain a feasible schedule. Further applying Lemma \ref{t:RoundRobin} we get that the makespan has to be $\sum_{j=1}^n p_j/m + \max\sett{P_{u'}}{{u'}\in[cm]} \leq LB + T \leq T + T = 2T$ as the optimal makespan has to be at least the lower bound and the largest $P_{u'}$ in the manipulated instance is also $T$.

Regarding the running time, we first compute all $P_u$ in time $\cO(n)$ as we have to add the processing time of each job exactly once. Then we compute the upper and lower bound in time $\cO(n)$. By Lemma \ref{t:BinarySearch} the advanced binary search admits $\cO(C \log(m))$ many iterations. For each guess $T$ of the makespan we first handle the large classes $P_u \geq T$ by distributing each class greedily onto a new, unique class until the accumulated processing times reach $T$. Then a job is cut if necessary and placed into the next new class. The last class may have a load of at most $T$. As for each class we have to touch each job exactly once to distribute it, we touch exactly $n$ jobs, thus this step is computable in time $\cO(n)$. If the amount of original classes with $P_u \leq T$ and new classes is larger than $c \cdot m$, we can discard the current guess $T$ as it cannot be a feasible schedule. Otherwise, we lower the guess. After finding the lowest feasible guess on $T$, we sort all classes in time $\cO(cm \log(cm))$. In the next step we allot each class as a whole onto the machines in time $\cO(cm)$ by round robin. Lastly, we reassign the correct classes via a simple mapping to the corresponding jobs, again in time $\cO(cm)$. This results in an overall running time of 
\begin{align*}
&\cO(n) + \cO(C\log(m)) \cdot \cO(n) + \cO(n) + \cO(cm + cm \log(cm)) \\
& = \cO(n) + \cO(Cn \log(m)) + \cO(cm \log(cm)) \\
& = \cO(n^2 \log(m)) + \cO(nm \log(nm))
\end{align*}
as $n \geq C$ and  $n \geq c$ holds. To achieve a fully polynomial algorithm regarding the number $n$ we have to make a small case distinction on the number of machines. If $m \leq n$ then our algorithm directly leads to the desired running time of $\cO(n^2 \log(n))$. However, in the splittable case the number of machines can be much larger, even exponential in the number of jobs. We can handle this obstacle by carefully reformulating some steps of the algorithm as shown next.

A large number of machines is first problematic when we sort all classes. It may appear that by introducing the new ones, we end up with $\cO(cm)$ many. But most of them will have size $T$ due to the slicing of the original jobs. Solely $C$ many classes will have size smaller than $T$, one for each original class. Thus instead of saving and sorting all of them, we only do so for the $C$ many with $P_{u'} < T$. For the remaining ones we just store their number. Now by applying round robin, we just distribute the $C$ classes onto $C \leq n < m$ machines and put an arbitrary class with $P_{u'} = T$ on top. For the remaining ones we just save the number of machines which are filled with two classes of size $T$. The space has to be sufficient, as we are only allowed to place two classes with $P_{u'} = T$ on top of each other or $T$ is a wrong guess. Hence the algorithm and the output only use $C$ classes and machine configurations explicitly, while $m$ only appears logarithmic. Thus this algorithm is fully polynomial in $n$, \ie
\begin{align*}
&\cO(n) + \cO(C\log(m)) \cdot \cO(n) + \cO(n) + \cO(C \log(C) + \log(m)) \\
& = \cO(n^2 \log(m))
\end{align*}
\end{proof}

 \subsection*{Preemptive Case}
In the preemptive case we are also allowed to cut the jobs and place them onto different machines. But in addition we have to take care that no job is scheduled in parallel, \ie at the same time on different machines. Indeed we can use the same algorithm here as for the one for the splittable case with just two slight adaptions taking care of the additional constraint. Firstly, the lower bound $LB$ has to guarantee that $T \geq p_{\max}$, such that there is enough space to schedule each job sequentially (on different machines). Thus we compute the lower bound as $LB = \max\{p_{\max}, \sum_{j=1}^{n} p_j/m \}$. Secondly, we have to reschedule some jobs to make sure, that no job piece is scheduled in parallel. Remark that this can only happen while dividing large classes into $C_u$ sub-classes if a job is cut at $T$ and the remaining piece is of size $P_{u'} < T$. Otherwise a sub-class with $P_{u'} = T$ will be assigned by round robin onto the bottom of a machine due to the lower bound and thus it will not collide with its other part being at the top. Hence, we only have to repack the machines, if there is at least one new class with $P_{u'} = T$. The modified code to repack the jobs is presented below and will be executed after (*) in the Algorithm \ref{alg:AlgoSplit}. Further an example showing such a repacking is visualized in Figure~\ref{fig:Repacking}. 

 \begin{center}
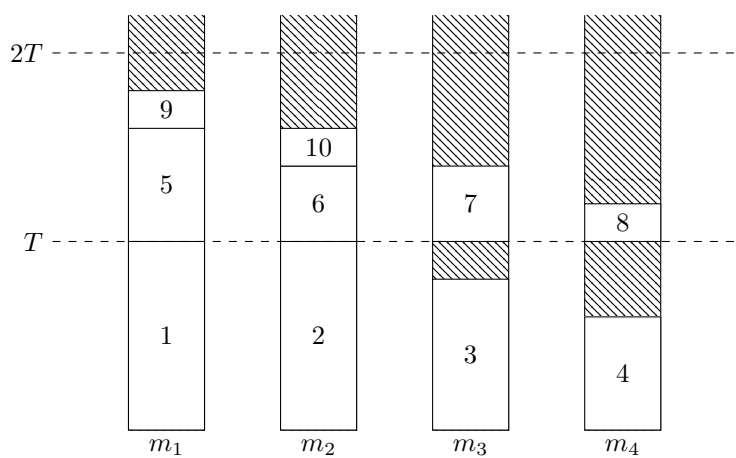
\begin{figure}
	\begin{tikzpicture}
	
	\pgfmathsetmacro{\w}{1}
	\pgfmathsetmacro{\h}{0.5}
	\pgfmathsetmacro{\mH}{11}
	
	\pgfmathsetmacro{\ww}{0}
	\drawMachine[$m_1$]{1*\w}{\mH*\h}{\ww*\w}{0};
	\pgfmathsetmacro{\hh}{0}
	\foreach \x/\z in {
		5/1,
		3/5,
		1/9
	}{
		\pgfmathsetmacro{\hhh}{\hh+\x}
		\drawJob[\z]{0*\w}{\hh *\h}{\w}{\hhh*\h};	
		\global\let\hh=\hhh
	}
	
	\pgfmathsetmacro{\ww}{2}
	\drawMachine[$m_2$]{1*\w}{\mH*\h}{\ww*\w}{0};
	\pgfmathsetmacro{\hh}{0}
	\foreach \x/\z in {
		5/2,
		2/6,
		1/10
	}{
		\pgfmathsetmacro{\hhh}{\hh+\x}
		\drawJob[\z]{\ww*\w}{\hh *\h}{\ww*\w +\w}{\hhh*\h};	
		\global\let\hh=\hhh
	}
	
	\pgfmathsetmacro{\ww}{4}
	\drawMachine[$m_3$]{1*\w}{\mH*\h}{\ww*\w}{0};
	\foreach \x/\xx/\z in {
		0/4/3,
		5/7/7
	}{
		\drawJob[\z]{\ww*\w}{\x *\h}{\ww*\w +\w}{\xx*\h};	
	}
	
	\pgfmathsetmacro{\ww}{6}
	\drawMachine[$m_4$]{1*\w}{\mH*\h}{\ww*\w}{0};
	\foreach \x/\xx/\z in {
		0/3/4,
		5/6/8
	}{
		\pgfmathsetmacro{\hhh}{\hh+\x}
		\drawJob[\z]{\ww*\w}{\x*\h}{\ww*\w +\w}{\xx*\h};	
		\global\let\hh=\hhh
	}

	\draw[dashed] (-1*\w,5*\h) node[left]{$T$} -- (8*\w,5*\h); 
	\draw[dashed] (-1*\w,10*\h) node[left]{$2T$} -- (8*\w,10*\h);
	\end{tikzpicture}
\caption{The figure shows the repacking for the preemptive case regarding the schedule of Figure \ref{fig:RoundRobin}. The approach is to shift the classes above the first one such that they start at $T$.} \label{fig:Repacking}
\end{figure}
\end{center}

\begin{algorithm}
\begin{center}
\fbox{\pseudocode[width = 9.5cm]{
\\
	\text{If there exists a class $i$ with $P_{u'} = T$:}\\
	\t[2] \text{For each machine:}\\
	\t[3] \text{Shift the schedule above the largest class, such that it starts at time $T$.} 
 	} }
 \end{center}
 \caption{Extension of the Algorithm \ref{alg:AlgoSplit} to be executed after (*) to solve the preemptive version of the CCS problem.}
 \end{algorithm}

 \begin{theorem}
The modified approximation algorithm above produces a solution $\sigma$ with $\mu(\sigma) \leq 2 \cdot \mu(\textsc{opt}(I))$ for the preemptive case in time  $\cO(n^2 \log(n))$.
 \end{theorem}
 \begin{proof}
 Since it was proven in the previous theorem that the algorithm produces a feasible solution bounded by $2T$, we now focus on the adaptions. As mentioned above, the new lower bound is $LB = \max\{p_{\max}, \sum_{j=1}^{n} p_j/m \}$ as apart from an equal distribution the makespan has also to be at least as large as the largest processing time, since we cannot process a job in parallel. This also implies, that each job will be cut by the algorithm at most once while introducing the new sub-classes.
 
The repacking will now guarantee that job parts of the same job will not be executed in parallel. As explained above, a job is only cut once at the height of $T$ if a class $u$ has $P_u > T$. Otherwise we do not have to repack the machines and the estimation from the splittable case holds. Thus assume we have a class $u$ with $P_u > T$, where by dividing it into $C_u$ new classes we cut one job. Now we assign the new classes via round robin and repack the machines. Afterwards we cannot have a collision between the parts of the same job, as either the cut job is placed at the bottom of a machine not intersecting with its other part at the top due to our definition of the lower bound. Or otherwise the job is scheduled above $T$, also not intersecting.

It remains to prove that while repacking we do not exceed a makespan of $2T$ by shifting jobs, which then start and thus end later. Suppose the contrary, we have a machine $m_k$ with a load $L_k > 2T$. Meaning, the classes have to be shifted and thus the largest class on that machine has a processing time smaller than $T$. Further the shifted jobs have a total processing time larger than $T$ to exceed $2T$ by starting at time $T$. As the first class on $m_k$ has a processing time smaller than $T$, the machine $m_\ell$ admitting a class with $P_{u'} = T$ has to be packed earlier via round robin. Thus the load on $m_\ell$ excluding the first class with processing time $T$ is also greater than $T$ as they have to be larger than the load on $m_k$ without the first class, see Lemma \ref{t:RoundRobin}. Thus we exceeded $2T$ on a machine without any gaps which is not possible by Lemma \ref{t:RoundRobin}.  

Regarding the running time it is not reasonable to have more machines than jobs as we have a lower bound of at least $p_{\max}$ and thus each job could simply be scheduled separately onto each machine obtaining an optimal schedule for all $m \geq n$ using only $n$ machines. Thus inserting $n$ instead of $m$ for the number of machines and different classes, we get a running time of
\begin{align*}
& \cO(n) + \cO(C\log(m)) \cdot  \cO(n) + \cO(nm \log(nm))\\
& = \cO(n^2 \log(n) + \cO(n^2 \log(n)) \\
& = \cO(n^2 \log(n))
\end{align*}
for the algorithm as the other parts stay the same regarding Theorem \ref{t:ConstSplit}. Remark that we could use a simple binary search instead of the advanced one here as the optimal makespan in the preemptive case is integral. However, it would not change the running time as a normal binary search would need $\cO(\log(n \cdot p_{\max})) = \cO(\log(cn2^n)) = \cO(n \log(n))$ many iterations as we can estimate the upper bound with each job of size $p_{\max}$ being placed in the same machine.
 \end{proof}
 
 \subsection*{Non-Preemptive Case}
Unlike the other cases we are not allowed to cut the jobs in the non-preemptive version of this problem. So, first, we also need a lower bound of $LB = \max\{p_{\max}, \sum_{j=1}^{n} p_j/m \}$ to guarantee that the largest job fits as well as the overall area is sufficient. Further handling the large classes is much trickier as we have to calculate a lower bound $C_u$ on the necessary class slots for class $u$ more carefully. Hence, we do not only compute the needed classes regarding the area, \ie  $C_{u}^1 = \lceil P_u / T \rceil$, we also compute a lower bound $C_u^2$ on needed class slots for large jobs with $p_j > (1/3)T$. Then we take the maximum of both values, \ie $C_u = \max\{C_u^1, C_u^2\}$. To compute $C_u^2$ we first count the number of jobs $k_u$ with $p_j > (1/2)T$. These jobs have to be on different machines to not exceed $T$ in an optimal schedule. On top of them we would like to place as many jobs with $(1/2) T \geq p_j > (1/3)T$ as possible. This problem can be solved greedily by assigning the largest fitting job of this size on top of the jobs with $p_j>(1/2)T$. Now dividing the number $\ell_u$ of still unassigned jobs with $p_j > (1/3)T$ by $2$ we get a lower bound on the number of machines needed to schedule them. Summing up, we compute $C_u^2 = k_u + \ell_u/2$. Finally, taking the maximum of the needed area $C_u^1$ and the needed number of machines $C_u^2$ to schedule the large jobs, we get $C_u$, a lower bound on the number of class slots used to schedule this class in an optimal schedule. Now we can divide the jobs of class $u$  into $C_u$ classes by using the LPT algorithm. Roughly speaking, LPT assigns the jobs in decreasing order regarding their processing time $p_j$, such that a job is placed onto the class with the lowest load at this moment. Further we will proceed with a standard binary search between the upper and lower bound, as the optimal makespan has to be integral and we cannot define the borders in the same way we did before. The rest of the algorithm stays the same yielding the following theorem:
 
 \begin{theorem} \label{t:ConstNPreem}
The approximation algorithm above with a modified lower bound, a new calculation of $C_u$ and LPT as a subroutine produces a solution $\sigma$ with $\mu(\sigma) \leq (7/3) \cdot \mu(\textsc{opt}(I))$ for the non-preemptive case in time $\cO(n^2 \log^2(n))$.
 \end{theorem}
 \begin{proof}
In Theorem \ref{t:ConstSplit} we already argued the correctness of the algorithm. Thus it remains to show, that the new calculation of $C_u$ and the LPT schedule as a subroutine are correct yielding the desired approximation ratio. When computing $C_u$  we first distribute the $k_u$ jobs with $p_j > (1/2)T$ over the classes. Clearly, if there is a packing with makespan $T$, then these jobs cannot be packed onto the same machine. Further, to compute a more precise lower bound on the number of needed class slots, we also distribute the remaining jobs with $p_j > (1/3)T$ onto as few machines as possible. Following the same argument as before, three of these jobs cannot be packed onto a single machine. Thus dividing these jobs by $2$ yields the correct number of needed classes. But some of these jobs could be packed on top of jobs with $p_j > (1/2)T$. Thus we first allot as many of them as possible by greedily assigning the largest fitting one onto the jobs with $p_j > (1/2)T$. This is optimal regarding the number of packed large jobs, as at most one fits on top. Furthermore, when we use the largest possible job, it could only be replaced by a smaller one. But exchanging these two jobs will only lead to a larger load on the other machine. Thus is can only worsen the distribution. Hence by our approach we pack as many jobs as possible on top of the jobs with $p_j > (1/2)T$. The remaining $\ell_u$ ones yield the number of additional machines, \ie class slots, to pack the large items of this class by dividing their number by $2$. Now taking the maximum of needed machines $C_u^2 = k_u + \ell_u/2$ to pack the large jobs and the area $C_u^1 = \lceil P_u / T \rceil$, we get a sufficient good lower bound on the number $C_u$ of needed class slots in an optimal schedule. Finally using the $LPT$ algorithm, we get a schedule with makespan at most $T + (1/3)T$, as the area is sufficient and we over pack each class with at most one job $j$ with $p_j \leq (1/3)T$. This can be seen, as the LPT schedule first distributes the jobs with $p_j > (1/2)T$ as they are the largest ones. These will be distributed over $C_u$ and as due to its calculation $C_u \geq k_u$. Then the jobs with $(1/2)T \geq p_j > (1/3)T$ are placed. Here we first distribute them onto the remaining empty machines. Then we place them on top of other already assigned jobs $j$ with $(1/2)T \geq p_j > (1/3)T$ and finally on top of the already assigned jobs processing time larger than $(1/2)T$. Since we calculated the number to assign them without over packing a machine, solely jobs with $p_j \leq (1/3)T$ will fill a machine above $T$. Further at most one of these jobs exceeding the makespan is packed onto each machine as otherwise the total area would not be sufficient. Proceeding with the algorithm we thus get an overall makespan of $\sum_{j=1}^n p_j/m + \max\sett{P_{u'}}{{u'}\in[cm]} \leq LB + (4/3)T \leq T+ (4/3)T = (7/3)T$ yielding the desired approximation ratio. 

Regarding the running time, calculating $C_u$ can be done in time $\cO(n \log(n))$ as we have to sort and then distribute the jobs of each class. Using LPT also takes time $\cO(n \log(n))$. The standard binary search takes time $\cO(\log(n \cdot p_{\max}) = \cO(\log(n2^2)) = \cO(n \log(n))$ due to our estimation of the lower bound. Since the other parts stay the same, we get an overall running time of $\cO(n) + \cO(n\log(n)) \cdot \cO(n \log(n)) + \cO(nm \log(nm)) = \cO(n^2 \log^2(n))$.
 \end{proof}

 
\section{PTAS for CCS} \label{CCS:PTAS}

We have seen how to solve different versions of the Class Constrainted problem efficiently and simple if we are satisfied with a quality of~2 or $7/3$ respectively. Now, we aim to improve this factor up to an additive term of $\epsilon \in (0, 1]$ for each case at the cost of more a complex procedure. We do so by introducing three PTASs each computing a solution $\sigma_P$ arbitrarily close to the optimal makespan, i.e., $\mu(\sigma_P) \leq (1+ \epsilon) \cdot \mu(\textsc{opt}(I))$. 
The results all have a similar structure: We assume that there is some accuracy parameter $\delta > 0$ with $1/\delta\in\ints$ depending on $\epsilon$ that is specified concretely for each sub-case. Assume that a guess $T$ on the optimal makespan is given. We design a procedure that computes a schedule with makespan $(1 + \cO(\delta))T$ or verifies that a solution with makespan $T$ does not exists. Embedding this in a binary search to find the correct guess on the makespan yields the PTAS. The idea to solve the problem for fixed guesses on the makespan instead of solving the minimization problem directly was introduced by Hochbaum and Shmoys \cite{hochbaum1987using}. 

In each of the cases, we first simplify the instance using grouping and rounding techniques.
In particular, we call a class $u \in [C]$ \emph{large} if each job of class $u$ has a processing time bigger than $\delta T$.
On the other hand, we call $u$ \emph{small} if there is exactly one job with class $u$ and this job has a processing time of at most $\delta T$. Grouping will then guarantee that each class $u$ is either small or large and we set $\xi_u = 1$ in the former and $\xi_u = 0$ in the latter case. Furthermore, rounding ensures that there are only few distinct processing times.
For the splittable and preemptive case, we have to prove additionally the existence of certain well-structured schedules. These can then be modeled using \nfold{} IPs. In the design of the \nfold{}s, we adapt and extend the so-called module configuration IP (MCIP) introduced in \cite{haessler1991cutting} to our needs. Finally, we have to prove that the solution of the \nfold{} indeed can be used to construct a feasible schedule with the desired approximation ratio.

Most effort and new ideas incorporate for these results are the techniques and observation for grouping the items accordingly and to prove the existence of certain well-structured schedules. Only then it was possible to design the corresponding \nfold{}s and finally to show that the solution of the IP can indeed be used to construct a feasible schedule with the desired approximation ratio. In the following, we consider the splittable, non-preemptive and preemptive case in this order.

\subsection{Splittable Case}
We begin with the splittable case as it appears to be the easiest of the three problems. However, there is some extra difficulty arising from the fact that the number of machines can be exponential. In that case we manage to lower the dependency on $m$ to logarithmic terms using some insights of the structure and extending our algorithm accordingly. But for now assume that $m$ can be bounded polynomial by $n$. The other case will be handled afterwards in a separate section. 

\subparagraph{Preprocessing.}
In the splittable case, we can simply group all jobs belonging to a class $u\in[C]$ into one job with processing time $p_u = \sum_{j\in J, c_j =u}p_j$.
If we have $p_u > \delta T$, the class $u$ is large and we set $\xi_u = 0$. Otherwise $u$ is small and we set $\xi_u = 1$.
It is easy to see that the problem is equivalent since the newly created jobs still can be split arbitrarily and behave the same concerning the class constraints.

We round up the processing times as follows. Let $j$ be a job. If $c_j$ is a large class, we set $p'_j = \ceil{p_j/(\delta^2 T)}\delta^2 T$. If not, $c_j$ is a small class and we set $p'_j = \ceil{p_j/(\delta^2 T/c)}\delta^2 T/c$.
Furthermore, we scale the makespan bound $T$ and the processing times by $c/(\delta^2 T)$ to ensure integral values, that is, we have $\delta^2 T/c  = 1$ afterwards. The resulting instance is denoted by $I'$, the job set by $J'$ and the processing time and class of job $j\in J'$ by $p'_j$ and $c'_j$, respectively. We also write $p'_u$ to denote the processing time of the single job of class $u\in[C]$.

\begin{lemma}
If there is a schedule with makespan $T$ for instance $I$, then there is also a schedule with makespan $(1+2\delta)T$ for $I'$.
\end{lemma}
\begin{proof}
A schedule with makespan $T$ for $I$ directly induces a schedule with the same makespan for the instance with the grouped jobs and reordered classes.
To realize the increased processing times, we may distribute the increase proportionally to its job pieces.
Note that the jobs belonging to large classes are increased at most by a factor of $(1+\delta)$ by the rounding procedure as each of them had a size of at least $\delta T$ before.
Hence, the load on each machine due to such jobs may increase at most by this factor. 
Furthermore, there can be at most $c$ job pieces belonging to small classes scheduled on each machine, and therefore the increase due to small jobs is upper bounded by $c\cdot \delta T/c = \delta T$.
\end{proof}

\subparagraph{Well-Structured Schedule.}
In the splittable case, we call a schedule well-structured if the following holds:
The size of each split piece of a job belonging to a large class is at least $\delta T$ and an integer multiple of $\delta^2 T$. 
Furthermore, jobs belonging to small classes are not split at all.
\begin{lemma}\label{lem:splittable_well-structured}
If there is a schedule with makespan $T'$ for instance $I'$, then there is also a well-structured schedule with makespan at most $T' + 2\delta T$ for $I'$.
\end{lemma}
\begin{proof}
Let there be a schedule with makespan $T'$ for instance $I'$ and $j\in J'$.
If $c_j$ is a large class, we divide $j$ into $n_j := \floor{p'_j/(\delta T)}$ many parts.
The size $s_{j,\ell}$ of the $\ell$-th part for each $\ell\in[n_j - 1]$ is $\delta T$ and we set $s_{j,n_j} = p'_j - \sum_{\ell = 1}^{n_j-1}s_{j,\ell} \in [\delta T, 2\delta T)$.
Note that due to the rounding, all the parts have a size that is an integer multiple of $\delta^2 T$.
The schedule for $j$ translates into a schedule for the job parts in a straight-forward fashion.
If $c_j$ is a small class, we have just one job part given by the whole job and therefore set $n_j = 1$ as well as $s_{j,1} = p'_j$.
Now, let $x^*_{(j,\ell), i}$ be the fraction of the $\ell$-th part of job $j$ that is assigned to machine $i$ in the given schedule.
Furthermore, let $z_{(j,\ell), i}\in \set{0,1}$ be equal to $1$ if some piece of the $\ell$-th part of job $j$ is assigned to machine $i$ and equal to $0$ otherwise.
It is easy to verify that $(x^*_{(j,\ell), i})$ is a feasible solution of the following LP:
\begin{align}
\sum_{j\in J',\ell\in[n_j]} s_{j,\ell}x_{(j,\ell),i} & \leq T'  & \forall i\in M \label{eq:splittable_structured_schedule_load}\\
\sum_{i\in M} x_{(j,\ell),i} & = 1  & \forall j\in J',\ell\in [n_j] \label{eq:splittable_structured_schedule_jobs}\\ 
0 \leq & x_{(j,\ell),i} \leq z_{(j,\ell), i} & \forall i\in M, j\in J',\ell\in [n_j]\label{eq:splittable_structured_schedule_bounds}
\end{align}
Employing a classical rounding result by Lenstra et al.~\cite{lenstra} yields a rounded solution $(\bar{x}_{(j,\ell), i})$ such that $\bar{x}_{(j,\ell), i}\in\set{0,1}$ holds, (\ref{eq:splittable_structured_schedule_jobs}) and (\ref{eq:splittable_structured_schedule_bounds}) are satisfied, and furthermore we have $\sum_{j\in J',\ell\in[n_j]} s_{j,\ell}\bar{x}_{(j,\ell),i}  \leq T' + \max_{j\in J',\ell\in[n_j]}s_{j,\ell} \leq T' + 2\delta T$ for each $i\in M$.
The rounded solution directly yields a well-structured schedule for $I'$ with makespan at most $T' + 2\delta T$.
\end{proof}

\subparagraph{Setting up the \nfold{}.}
Taking the above steps into considerations, we set $\bar{T} = (1+4\delta) =(1+\cO(\delta))T$ and search for a well-structured schedule with makespan $\bar{T}$ via an \nfold{} IP.
Following the module configuration framework \cite{haessler1991cutting} in which \emph{modules} are used to cover the basic objects we design our \nfold{}. That is, basic objects correspond to jobs and configurations in turn are used as modules.
In this context, we define the set of modules $\mathcal{M}$ to be the set of possible split sizes of jobs from large classes in a well-structured schedule with makespan $\bar{T}$, that is, $\mathcal{M} = \sett[\big]{\ell \delta^2 T}{ \ell \in\set{1/\delta,\dots,\bar{T}/(\delta^2 T)}}$.
A configuration $K\in \ints_{\geq 0}^{\mathcal{M}}$ is a multiplicity vector of modules and its size $\Lambda(K)$ is given by $\sum_{q\in\mathcal{M}} K_q q$.
Intuitively, each module in a configuration should belong to a distinct class and $\norm{K}_1 = \sum_{q\in\mathcal{M}} K_q$ corresponds to the number of class slots used in the configuration.
We consider the set $\mathcal{K}$ of configurations $K$ with $\sum_{q\in\mathcal{M}} K_q q \leq \bar{T}$ and $\norm{K}_1 \leq c$ and denote the set of configuration sizes as $\Lambda(\mathcal{K})$.
Let $K\in \mathcal{K}$. 
Because $q\geq \delta T$ for each $q\in\mathcal{M}$, we know that $\norm{K}_1\leq \bar{T}/\delta T =\cO(1/\delta) $.
We set $c^* = \min\set{\bar{T}/\delta T, c}$, and for each $h\in\Lambda(\mathcal{K})$ and $b\in[c^*]$, we define $\mathcal{K}(h,b) = \sett{K\in\mathcal{K}}{\Lambda(K) = h, \norm{K}_1 = b}$.

We introduce three types of variables each of which is duplicated for each class $u\in [C]$ corresponding to the blocks of the \nfold{}. Remark that the duplication has no meaning itself as we always consider all duplicates simultaneously. It is solely used to obtain the \nfold{} structure. Let $u\in[C]$ be a class.
We have a variable $y^u_q\in\set{0,\dots, m\bar{T}/(\delta T)}$ for each module $q\in \mathcal{M}$ indicating how often $q$ is chosen to cover the job of class $u$.
Moreover, we introduce a variable $x_K^u \in\set{0,\dots, m}$ for each configuration $K\in \mathcal{K}$. 
We use the above variables to handle the assignment of large classes.
To deal with the small classes, we have binary variables $ z^u_{h,b}\in\set{0,1}$ for each $h\in\Lambda(\mathcal{K})$ and $b\in[c^*]$ which are used to decide whether the class is assigned to a machine on which $b$ job pieces with overall size $h$ belonging to large classes are scheduled.
The actual schedule of the small classes will be determined using the round robin procedure, which we described in the previous section, for each size $h$ and number of class slots $b$.
The \nfold{} has the following constraints:
 \begin{align*}
\tag{0}
\label{c0}
&\sum\limits_{u=1}^C \sum\limits_{K \in \mathcal K} x^u_K = m\\
\tag{1}
\label{c1}
&\sum\limits_{u=1}^C \sum\limits_{K \in \mathcal K} K_q  x^u_K = \sum\limits_{u=1}^C  y^u_{q} & \forall q \in \mathcal M \\
\tag{2}
\label{c2}
&\sum\limits_{u=1}^{C} z^u_{h,b} + b  \sum\limits_{u=1}^C \sum\limits_{K \in \mathcal K(h,b)} x^u_K \leq c  \sum\limits_{u=1}^C \sum\limits_{K \in \mathcal K(h,b)} x^u_K &\forall h \in \Lambda(\mathcal{K}),b \in [c^*]\\
\tag{3}
\label{c3}
&\sum\limits_{u=1}^{C} p'_u  z^u_{h,b} + h  \sum\limits_{u=1}^C \sum\limits_{K \in \mathcal K(h,b)} x^u_K \leq \bar{T} \sum\limits_{u=1}^C \sum\limits_{K \in \mathcal K(h,b)} x^u_K &\forall h \in \Lambda(\mathcal{K}),b \in [c^*]\\
\tag{4}
\label{c4}
&\sum\limits_{q\in\mathcal{M}} q y^u_q = (1 - \xi_u)p'_u &\forall u \in [C]\\
\tag{5}
\label{c5}
&\sum\limits_{h\in \Lambda(\mathcal{K})} \sum\limits_{b\in [c^*]} z^u_{h,b} = \xi_u & \forall u \in [C]
\end{align*}
Constraint~(\ref{c0}) guarantees that we choose the correct number of configurations.
The next constraint, namely Constraint~(\ref{c1}) , is satisfied if the chosen configurations cover the chosen modules; and due to Constraint~(\ref{c4}) the chosen modules cover the job of a class if that class is large.
If a class is small, on the other hand, Constraint~(\ref{c4}) and (\ref{c5}) ensure that no modules are chosen for this class and that the job of this class is assigned to exactly one type of configuration, respectively.
Lastly, Constraint~(\ref{c2}) and (\ref{c3}) make sure that there is a proper amount of space and class slots for the small classes left.
It is easy to see that the last two constraints are locally uniform and the remaining ones are globally uniform.
To apply Theorem \ref{nfold}, we still have to state in which direction we aim to optimize. 
However, as we are only aiming for a feasible makespan for the given guess $T$, we can set the objective function to zero. 

\begin{lemma}\label{lem:splittable_nfold_correct}
If there is a well-structured schedule with makespan $\bar{T}$ for instance $I'$, then there is also a solution to the above \nfold{} IP.
\end{lemma}
\begin{proof}
Given a well-structured schedule, the size of each job piece belonging to a large class and scheduled on any machine is included in $\mathcal{M}$.
Hence, for each machine we may count for each possible size the number of present pieces and thereby derive a configuration.
We set the $x^1$-variables accordingly, and set the $x^u$ variables for $u\neq 1$ to $0$.
Let $u\in[c^*]$.
If $u$ is a large class, it is split into pieces with sizes included in $\mathcal{M}$ for the schedule.
We set the $y^u$ variables accordingly and the $z^u$ variables to $0$.
If, on the other hand, $u$ is a small class, then the whole class is scheduled on the same machine $i$.
Let $K$ be the configuration corresponding to $i$, $h = \Lambda(K)$, and $b = \norm{K}_1$.
We set $z^u_{h,b} = 1$ and $z^u_{h',b'} = 0$ for each $(h',b')\neq (h,b)$.
Furthermore, we set all the $y^u$ variables to $0$.
It is easy to verify that this solution is feasible.
\end{proof}
Hence, if the \nfold{} has no feasible solution, we can reject the makespan guess $T$.

\subparagraph{Solving the \nfold{}.}
Now, we can make use of Theorem \ref{nfold} to solve the given \nfold{} Inter Linear Program. 
To estimate the running time, we have to bound the parameters $r, s, t$, $\Delta$ and $L$.
First note that $|\mathcal{M}| = \cO(1/\delta^2)$, $c^* = \cO(1/\delta) = |\Lambda(\mathcal{K})|$ and $|\mathcal{K}| \leq (|\mathcal{M}| + 1)^{\cO(1/\delta)} = 2^{\cO(1/\delta\log(1/\delta))}$.
Hence, we have $r = 1 + |\mathcal{M}| + c^*\times|\Lambda(\mathcal{K})| = \cO(1/\delta^2)$ many globally uniform constraints and $s = 2$ many locally uniform constraints.
Furthermore, concerning the brick size $t$, we have  $t = |\mathcal{K}| + |\mathcal{M}| + 3c^*|\Lambda(\mathcal{K})| + 2 = 2^{\cO(1/\delta\log(1/\delta))}$, taking into account the introduction of slack variables in each brick to transform (\ref{c2}) and (\ref{c3}) into equality conditions.
The absolute value of each number in the constraint matrix is upper bounded by $\bar{T}$ and due to the scaling we have $\bar{T} = \cO(c/\delta^2)$.
Hence, for the largest number $\Delta$, we have $\Delta =\cO(c/\delta^2)$.
Lastly, we have to estimate the encoding length $L$ of the largest number in the input.
Again, numbers with absolute value $\cO(c/\delta^2)$ as well as multiples of the machine number $m$ as big as $\cO(m/\delta)$ may occur (the upper bounds of the $y$ variables).
Therefore, we have $L = \cO(\log(m/\delta + c/\delta^2))$.
Summing up, we get the following running time (using $C\geq \max\set{2,c}$):
\[(rs\Delta)^{\cO(r^2s + s^2)} L \cdot Nt \log^{\cO(1)}(Nt) \leq C^{\cO(1/\delta^4\log(1/\delta))} \cdot \log (m)\]

\subparagraph{Constructing the Schedule.}
Given this solution, we still have to build the schedule. 
For each large class $u$, we split the job of class $u$ into $y_q^u$ pieces of size $q$ for each $q\in\mathcal{M}$.
Next, we assign the configurations chosen by the $x$-variables onto the machines.
Given a machine with configuration $K$, we create $K_q$ slots of size $q$ for each $q\in\mathcal{M}$.
Then, we assign the job pieces greedily into fitting slots on the machine.
It is easy to see that these steps are successful due to the constraints of the \nfold{}. 
Lastly, we have to assign the small classes.
To do so, we again employ the round robin approach:
For each $h\in\Lambda(\mathcal{K})$ and $b\in[c^*]$, we assign the jobs of the small classes $u$ with $z_{h,b}^u = 1$ onto the machines with configurations $K\in\mathcal{K}(h,b)$ via round robin.
Due to (\ref{c2}), all the jobs can be placed by this procedure.
Furthermore, due to Lemma \ref{t:RoundRobin} and (\ref{c3}), this yields a schedule with makespan at most $\bar{T} + \delta T$.
Lastly, we have to use the original running times and jobs, which can be done using a greedy approach.

The overall running time for placing the large classes is linear in the number of involved job pieces, that is $\cO(m/\delta)$.
When placing the small jobs, we touch each class at most once. Further, when we insert the original jobs and job sizes, we have to consider each job and job piece in the schedule once.
The overall running time can thus be bounded by $\cO(n(m/\delta + C))$.

\subparagraph{Total running time and error.}
We have seen how to solve the  splittable version of the Class Constrained Scheduling problem, when we are given a guess $T$ on the makespan.  
Indeed, this requires a binary search for the optimal makespan, which can be done in time $\cO(\log((n \cdot p_{\max}/\delta))$, as $n \cdot p_{\max}$ is an upper bound on the largest possible makespan and we allow an error of $\cO(\delta)$. 
Using the (reasonable) assumption $ C\leq n$, we get a total running time of:
\begin{align*}
& \underbrace{\cO(\log(n \cdot p_{\max}/\delta))}_\text{Binary Search} \cdot [\underbrace{\cO(n)}_\text{Preprocess} + \underbrace{ C^{\cO(1/\delta^4\log(1/\delta))} \cdot \log (m)}_\text{\nfold{}} + \underbrace{\cO(n(m/\delta + C))}_\text{Constructing the packing}] \\
= & n^{\cO(1/\delta^4\log(1/\delta))} m\log (m) \log(p_{\max})
\end{align*}
Furthermore, the error in every phase can be bounded by $\cO(\delta)$ as we analyzed above. Thus the overall error is given by $\cO(\delta)$. Setting $\epsilon = \cO(\delta)$ we get the desired approximation ratio. This yields the total algorithm and analysis for the problem. 
The next theorem summarizes the results.

\begin{theorem}
A schedule $\sigma$ for the splittable version of the Class Constrained Scheduling problem is obtained in time $n^{\cO(1/\epsilon^4\log(1/\epsilon))}m \log (m) \log(p_{\max})$ with makespan $\mu(\sigma_P) \leq (1+ \epsilon) \cdot \mu(\textsc{opt}(I))$, where OPT$(I)$ denotes a solution with optimal makespan for packing the instance $I$ when $m$ can be bounded by a polynomial in $n$. This yields the desired PTAS $P$ for this problem.
\end{theorem}

\subparagraph{Handling an Exponential Number of Machines.}

Again we can have the problem that $m$ can be exponentially large in the number of jobs. 
Then the algorithm described above would not compute a solution in polynomial time regarding~$n$. However, we can handle this by extending our algorithm using a simple idea from \cite{DBLP:conf/innovations/JansenKMR19}. 
 
First, observe that we can convert any schedule into a schedule in which each machine has that same load and each pair of classes occurs on at most one machine.
Indeed, if we have two machines $i_1$ and $i_2$ on which the same pair $(u_1,u_2)$ occurs, we can apply a simple swap.
Let $p(i,u)$ be the overall load of a class $u\in\set{u_1,u_2}$ on a machine $i\in\set{i_1,i_2}$.
\Wlogeneral we may assume that $p(i_1,u_1)$ is minimal.
We move all the job pieces of class $u_1$ placed on machine $i_1$ to machine $i_2$ and job pieces of class $u_2$ with overall size $ p(i_1,u_1)$ from machine $i_2$ to $i_1$.
Afterwards, both machines have the same load, and class $i_1$ does not occur on machine $i_1$.
Moreover, the number of used class slots has not increased on any machine. 
The approach is visualized in Figure \ref{fig:ChangJobs}.
\begin{figure}
\begin{center}
 	\begin{tikzpicture}
	
	\pgfmathsetmacro{\w}{1}
	\pgfmathsetmacro{\h}{0.5}
	\pgfmathsetmacro{\mH}{8}
	
	\pgfmathsetmacro{\ww}{0}
	\drawMachine[$m_1$]{1*\w}{\mH*\h}{\ww*\w}{0};
	\pgfmathsetmacro{\hh}{0}
	\foreach \x/\xx/\z in {
		0/1/$ $,
		1/3/1,
		3/4/$ $,
		4/6/2
	}{
		\pgfmathsetmacro{\hhh}{\hh+\x}
		\drawJob[\z]{\ww*\w}{\x*\h}{\ww*\w +\w}{\xx*\h};	
		\global\let\hh=\hhh
	}
	
	\draw[dashed] (1*\w,1*\h) -- (2*\w,1*\h);
	\draw[dashed] (1*\w,2*\h) -- (2*\w,2*\h);
	
	\pgfmathsetmacro{\ww}{2}
	\drawMachine[$m_2$]{1*\w}{\mH*\h}{\ww*\w}{0};
	\pgfmathsetmacro{\hh}{0}
	\foreach \x/\xx/\z in {
		0/1/$ $,
		1/2/2,
		2/4/$ $,
		4/7/1
	}{
		\pgfmathsetmacro{\hhh}{\hh+\x}
		\drawJob[\z]{\ww*\w}{\x*\h}{\ww*\w +\w}{\xx*\h};	
		\global\let\hh=\hhh
	}
	
	\node[scale=2] at (4.5*\w,4*\h) {$\Rightarrow$};
	
	\pgfmathsetmacro{\ww}{6}
	\drawMachine[$m_1$]{1*\w}{\mH*\h}{\ww*\w}{0};
	\foreach \x/\xx/\z in {
		0/1/$ $,
		1/2/1,
		2/3/2,
		3/4/$ $,
		4/6/2
	}{
		\pgfmathsetmacro{\hhh}{\hh+\x}
		\drawJob[\z]{\ww*\w}{\x*\h}{\ww*\w +\w}{\xx*\h};	
		\global\let\hh=\hhh
	}
	
	\pgfmathsetmacro{\ww}{8}
	\drawMachine[$m_2$]{1*\w}{\mH*\h}{\ww*\w}{0};
	\foreach \x/\xx/\z in {
		0/1/$ $,
		1/2/1,
		2/4/$ $,
		4/7/1
	}{
		\drawJob[\z]{\ww*\w}{\x *\h}{\ww*\w +\w}{\xx*\h};	
	}

	\end{tikzpicture}
 \caption{This figure shows the exchange of two job pieces, such that the machines have distinct pairs of classes.} \label{fig:ChangJobs}
 \end{center} 
\end{figure}
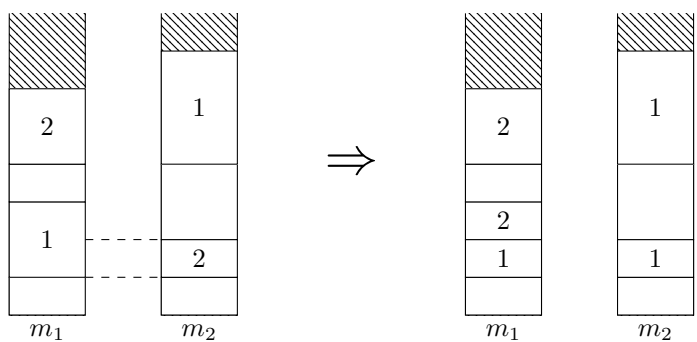
Furthermore, in a second step, we can transform the schedule even further.
For each class $u$, we can guarantee that there is at most one machine that exclusively executes pieces belonging to $u$ and is not fully filled.
Again, this is realized via a simple swapping argument. 

We can modify the \nfold{} correspondingly:
There two configurations that we call trivial, namely the one that that chooses the largest module exactly once, and the one which does not choose any module.
Let $\mathcal K'$ be the subset of non-trivial configurations $\mathcal K'$.
Because of the above considerations, we may introduce the following globally uniform constraint without violating Lemma \ref{lem:splittable_nfold_correct}:
 \begin{align*}
&\sum\limits_{u=1}^C \sum\limits_{K \in \mathcal K'} x_K^u \leq {{C}\choose{2}} + C
\end{align*}
It is easy to verify that the increase in the running time vanishes in the $\cO$-Notation. 

Now, when constructing the packing, we first can deal with the trivial configurations and remember for each class the corresponding number of machines that are fully filled with these class.
The sizes of the classes are decreased accordingly and in the following only the at most ${{C}\choose{2}} + C$ machines are taken into account.
Besides this, the algorithm stays the same.

This approach yields an improved running time of $n^{\cO(1/\epsilon^4\log(1/\epsilon))} \log (m) \log(p_{\max})$.

\begin{theorem}
A schedule $\sigma$ for the splittable version of the Class Constrained Scheduling problem is obtained in time $n^{\cO(1/\epsilon^4\log(1/\epsilon))}\log (m) \log(p_{\max})$ with makespan $\mu(\sigma_P) \leq (1+ \epsilon) \cdot \mu(\textsc{opt}(I))$, where OPT$(I)$ denotes a solution with optimal makespan for packing the instance $I$ when $m$ cannot be bounded by a polynomial in $n$. This yields the desired PTAS $P$ for this problem.
\end{theorem}


\subsection{Non-Preemptive Case}
For most parts, the non-preemptive case works quite similar to the splittable one. However, we have to consider that the jobs have to be assigned as a whole. Hence, we can not simply glue jobs from the same class together and thus forgetting their original structure. Instead, we group the jobs appropriately and later on look at each grouped job separately by defining modules, module sizes, and configurations more carefully. Adapting the remaining steps, we get the desired PTAS. Note that regarding the grouping a similar approach has been used in \cite{shachnai2001polynomial}.

\subparagraph{Preprocessing.}
First, we construct an instance $I'$ in which each class is either small or large by grouping the jobs. For each class $u$, we perform the following steps:
\begin{itemize}
\item As long as it is possible, repeatedly perform the following steps: Select a set of jobs $X \subseteq\sett{j\in J}{c_j = u, p_j < \delta T} $ such that $p(X) \in [\delta T, 2\delta T)$; remove $X$; introduce a new job with class $u$ and size $p(X)$.
\item Let $Y = \sett{j\in J}{c_j = u, p_j < \delta T}$. We have $p(Y)< \delta T$ because of the above step. 
\item If $u$ contains jobs not belonging to $Y$, we pick such a job $j$, remove $j$ and $Y$ from the instance, and introduce a new job of class $u$ with size $p_j + p(Y)$.
\item Otherwise, we remove $Y$ from the instance and introduce a new job of class $u$ and with size $p(Y)$.
\end{itemize}
Note that each of the newly created jobs has a size of at most $3\delta T$ and that there are indeed only small and large classes left. 
If a class $u$ is small, we set $\xi_u = 1$. Otherwise we set $\xi_u = 0$.
We call the resulting instance $I'$, the corresponding set of jobs $J'$, and write $p'_j$ and $c'_j$, respectively, to denote the processing time and class of job $j\in J'$. 
Furthermore, the following holds:
\begin{lemma}\label{lem:preprocess_non-preemptive}
If there is a schedule with makespan $T$ for $I$, then there is also a schedule with makepan $(1+3\delta)T$ for instance $I'$.
\end{lemma}
\begin{proof}
There is a set $J(j')\subseteq J\setminus J'$ for each newly introduced job $j'\in J'\setminus J$ such that $p'_{j'} = \sum_{j\in J(j')}p_j$ and $\sett{J(j')}{j'\in J'\setminus J}$ is a partition of $J\setminus J'$.
Given a schedule for $I$ with makespan $T$, let $y_{i,j} \in \set{0,1}$ be equal to $1$ if $j\in J$ is executed on machine $i$ and equal to $0$ otherwise.
We set $T_i = \sum_{j\in J\setminus J'} y_{ij}p_j$ for each machine $i$, $x^*_{ij'} = (\sum_{j\in J(j')}p_j y_{ij})/p_{j'}$ for each job $j'\in J'\setminus J$, and $z_{ij'} = \ceil{x^*_{ij'}}$.
Note that $x^*_{ij'} \in [0,1]$ and $z_{ij'} \in\set{0,1}$.
It is easy to see that $(x^*_{ij'})$ is a feasible solution of the following LP:
\begin{align}
\sum_{j'\in J'\setminus J} p'_{j'}x_{ij'} & \leq T_i  & \forall i\in M \\
\sum_{i\in M} x_{ij'} & = 1  & \forall j'\in J'\setminus J \label{eq:LP_preprocess_non_preemptive_1}\\ 
0 \leq & x_{ij'} \leq z_{ij'} & \forall i\in M, j'\in J'\setminus J\label{eq:LP_preprocess_non_preemptive_2}
\end{align}
Similar to the proof of Lemma \ref{lem:splittable_well-structured}, we can employ the classical result by Lenstra et al.~\cite{lenstra} to get a rounded solution $(\bar{x}_{i,j'})$ such that $\bar{x}_{i,j'}\in\set{0,1}$ holds, (\ref{eq:LP_preprocess_non_preemptive_1}) and (\ref{eq:LP_preprocess_non_preemptive_2}) are satisfied, and furthermore we have $\sum_{j'\in J'\setminus J} p'_{j'}\bar{x}_{i,j'}  \leq T_i + \max_{j'\in J'\setminus J}p'_{j'} \leq T_i + 3\delta T$ for each $i\in M$.

Hence, we can generate a suitable schedule by removing the jobs belonging to $J\setminus J'$ and assigning the jobs belonging to $J'\setminus J$ based on the $\bar{x}$-variables.
\end{proof}

Lastly, we round and scale the processing times and the makespan like in the splittable case and call the resulting instance $I''$.
For small classes $u$, we write $p''_u$ to denote the processing time of the single job of the class.
Furthermore, we denote the set of rounded processing times occurring in large classes by $\mathcal{P}$ and the number of jobs of class $u$ and size $p\in\mathcal{P}$ by $n_p^u$.

\subparagraph{Setting up the \nfold{}.}

Considering the error produced by the preprocessing, we set  $\bar{T} = (1+3\delta)(1+2\delta) =(1+\cO(\delta))T$. 
As for this case, modules are defined as multiplicity vectors of processing times, i.e., $\mathcal{M} = \sett{M\in\ints_{\geq 0}^\mathcal{P}}{\sum_{p\in\mathcal{P}}M_p p \leq \bar{T}}$. Indeed the modules are similar to the configurations in the splittable case. However, the modules represent the packing of distinct jobs of a single class on some machine instead of solely stating the volume of the class. The size $\Lambda(M)$ of a module $M$ is given by $\sum_{p\in\mathcal{P}}M_p p$ and the set of module sizes is denoted as $\Lambda(\mathcal{M})$.
Note that in the splittable case we did not distinguish a module and its sizes and considered configurations of module sizes.
In contrast, we define configurations for this case as multiplicity vectors of module \emph{sizes} $K\in\ints_{\geq 0}^{\Lambda(\mathcal{M})}$, and the size $\Lambda(K)$ of a configuration $K$ is given by $\sum_{q\in\Lambda(\mathcal{M})}K_q q$.
The set of configurations $\mathcal{K}$ is given by the configurations $K$ with $\Lambda(K) \leq \bar{T}$ and $\norm{K}_1\leq c$. Like before, $\Lambda(\mathcal{K})$ is the set of configuration sizes occurring in $\mathcal{K}$, $c^* = \min\set{\bar{T}/(\delta T),c}$, and $\mathcal{K}(h,b) = \sett{K\in\mathcal{K}}{\Lambda(K) = h, \lVert K\rVert_1 =b}$.

Let $u\in[C]$ be a class.
We have a variable $y^u_M\in\set{0,\dots, m}$ for each module $M\in \mathcal{M}$ indicating how often $M$ is chosen to cover the jobs of class $u$.
Moreover, we introduce a variable $x_K^u \in\set{0,\dots, m}$ for each configuration $K\in \mathcal{K}$. 
Like before, the duplication of these variables does not carry meaning and is solely used to obtain the \nfold{}-structure. 
Furthermore, we have binary variables $ z^u_{h,b}\in\set{0,1}$ for each $h\in\Lambda(\mathcal{K})$ and $b\in[c^*]$ which are used to decide whether the class is assigned to a machine on which $b$ job pieces with overall size $h$ belonging to large classes are scheduled.
The \nfold{} has the following constraints:

\begin{align*}
\tag{0}
\label{c0_np}
&\sum\limits_{u=1}^C \sum\limits_{K \in \mathcal K} x^u_K = m\\
\tag{1}
\label{c1_np}
&\sum\limits_{u=1}^C \sum\limits_{K \in \mathcal K} K_q  x^u_K = \sum\limits_{u=1}^C \sum\limits_{M\in\mathcal{M}:\Lambda(M) = q} y^u_{M} & \forall q \in \Lambda(\mathcal M)\\
\tag{2}
\label{c2_np}
&\sum\limits_{u=1}^{C} z^u_{h,b} + b  \sum\limits_{u=1}^C \sum\limits_{K \in \mathcal K(h,b)} x^u_K \leq c  \sum\limits_{u=1}^C \sum\limits_{K \in \mathcal K(h,b)} x^u_K &\forall h \in \Lambda(\mathcal{K}),b \in [c^*]\\
\tag{3}
\label{c3_np}
&\sum\limits_{u=1}^{C} p''_u  z^u_{h,b} + h  \sum\limits_{u=1}^C \sum\limits_{K \in \mathcal K(h,b)} x^u_K \leq \bar{T} \sum\limits_{u=1}^C \sum\limits_{K \in \mathcal K(h,b)} x^u_K &\forall h \in \Lambda(\mathcal{K}),b \in [c^*]\\
\tag{4}
\label{c4_np}
&\sum\limits_{M\in\mathcal{M}} M_p y^u_M = (1 - \xi_u)n_p^u &\forall u \in [C],p\in \mathcal{P}\\
\tag{5}
\label{c5_np}
&\sum\limits_{h\in \Lambda(\mathcal{K})} \sum\limits_{b\in [c^*]} z^u_{h,b} = \xi_u & \forall u \in [C]
\end{align*}

Note that the \nfold{} is very similar to the one used in the splittable case. Nevertheless, the Constraints (\ref{c1_np}) and (\ref{c4_np}) had to be adjusted to deal with the changed definitions of modules and configurations. In detail, Constraint~(\ref{c0_np}) guarantees that the number of configurations matches the number of available machines $m$. The second constraint~(\ref{c1_np}) is satisfied if the chosen configurations cover the modules. Similarly, (\ref{c4_np}) assures that the chosen modules indeed cover the large jobs. For the small classes, (\ref{c4_np}) ensures that the job is not covered. Instead, Constraint~(\ref{c5_np}) watches that a small job is assigned to exactly one configuration. Finally, Constraint~(\ref{c2_np}) and (\ref{c3_np}) assure that the area and the number of class slots is sufficient for the small classes. Again only (\ref{c4_np}) and (\ref{c5_np}) are locally uniform. Lastly we set the objective function to zero as we are solely aiming for a feasible solution.

\begin{lemma}\label{lem:non-preemptive_nfold_correct}
If there is a schedule with makespan $\bar{T}$ for instance $I''$, then there is also a solution to the above \nfold{} IP.
\end{lemma}
\begin{proof}
Given a well-structured schedule, each subset of jobs belonging to a large class is included in $\mathcal{M}$ having a specific size in $\Lambda(M)$. Hence, for each machine we may count for each possible size the number of present module sizes and thereby derive a configuration. We set the $x^1_K$ variables accordingly, and set the $x^u$ variables for $u\neq 1$ to $0$.
Let $u\in[c^*]$. If $u$ is a large class, its jobs may be distributed along the machines whereas on each machine the jobs of a class build up a module. Thus we can count the multiplicity of each used module for a certain class and set the $y^u_M$ variables accordingly and the $z^u$ variables to $0$. If, on the other hand, $u$ is a small class, then it contains only one job which is scheduled on exactly one machine $i$. Let $K$ be the configuration corresponding to $i$, $h = \Lambda(K)$, and $b = \norm{K}_1$. We set $z^u_{h,b} = 1$ and $z^u_{h',b'} = 0$ for each $(h',b')\neq (h,b)$. Furthermore, we set all the $y^u$ variables to $0$.
It is easy to verify that this solution is feasible.
\end{proof}
Hence, if the \nfold{} has no feasible solution, we can reject the makespan guess $T$. 

\subparagraph{Solving the \nfold{}.}

Again, we have to bound the parameters $r, s, t$, $\Delta$ and $L$ in the application of Theorem \ref{nfold}.
We have $s= |\mathcal{P}| + 1 = \cO(1/\delta^2)$ locally uniform constraints due to the rounding of the processing times.
Moreover, there are $r = 1 + |\Lambda(\mathcal{M})| + 2|\Lambda(\mathcal{K})|c^*$ globally uniform constraints.
Note that $|\mathcal{M}| = (|\mathcal{P}| + 1)^{\cO(1/\delta)} = 2^{\cO(1/\delta\log(1/\delta)}$, $|\Lambda(\mathcal{M})| = \cO(1/\delta^2)$, $|\mathcal{K}| = (|\mathcal{M}| + 1)^{\cO(1/\delta)} = 2^{\cO(1/\delta\log(1/\delta)}$, $|\Lambda(\mathcal{K})| = \cO(1/\delta^2)$, and $c^* = \cO(1/\delta)$.
Hence, we have $r = \cO(1/\delta^3)$.
There are $t = |\mathcal{K}| + |\mathcal{M}| + 3|\Lambda(\mathcal{K})|c^*$ many variables for each block (including slack variables for Constraint (\ref{c2_np}) and (\ref{c3_np})), and therefore we have $t = 2^{\cO(1/\delta\log(1/\delta)}$.
Like in the splittable case, the largest number $\Delta$ can be upper bounded by $\cO(c/\delta^2)$.
Concerning the encoding length $L$ of the largest number in the input, we additionally have to take the upper bounds of the variables into account yielding $L=\cO(\log(c/\delta^2 + m))$.
Summing up, we get the following running time (using $C\geq \max\set{2,c}$):
\[(rs\Delta)^{\cO(r^2s + s^2)} L \cdot Nt \log^{\cO(1)}(Nt) \leq C^{\cO(1/\delta^8\log(1/\delta))} \cdot \log m\]

\subparagraph{Constructing the Schedule.}

Having this solution at hand, we still have to build the schedule. 
We assign the configurations chosen by the $x$-variables onto the machines. 
Now we unfold the configurations as  Figure \ref{fig:unfoldingConfig} shows. 
In detail, given a machine with configuration $K$ we create $\norm{K}_1$ slots where exactly $K_q$ slots have size $q$ for each $q \in \Lambda(\mathcal M)$. 
Each slot is then greedily filled with a module corresponding to the assignment of $y^u_M$. 
Next, each module $M$ is dissolved in the corresponding multiplicities $M_p$ of job lengths $p$. 
Then we assign the corresponding jobs greedily into the slots of the job lengths. 
Afterwards all large jobs are allotted. 
Due to the constraints of the ILP, it is easy to see that these steps are successful. 
Regarding the small jobs, we distribute them like we did in the splittable case: 
For each $h\in\Lambda(\mathcal{K})$ and $b\in[c^*]$, we assign the jobs of the small classes $u$ with $z_{h,b}^u = 1$ onto the machines with configurations $K\in\mathcal{K}(h,b)$ via round robin. 
Due to Lemma \ref{t:RoundRobin} and (\ref{c3}), this yields a schedule with makespan at most $\bar{T} + \delta T$. 
In the last step we have to reinsert the original, non-rounded processing times and jobs. 

\begin{figure}
\begin{center}

\begin{tikzpicture}[stack/.style={
  rectangle split, rectangle split parts=#1, draw, anchor=center},
  myarrow/.style={single arrow, draw=none}]

\node [stack=4, name=K] (K)  {$K_1=3$\nodepart{two}$K_2=2$%
   \nodepart{three}$K_3=3$\nodepart{four}$K_4=0$};

\node [stack=2,right=of K, name=M1] (M1) {$L$\nodepart{two}$M$};
  
  \node [stack=4,right=of M1] (M2) {$M_1=0$\nodepart{two}$M_2=1$
  \nodepart{three}$M_3=1$\nodepart{four}$M_4=3$};
  
  \node [stack=1,right=of M2] (J) {$J_2^4$};

\node [above=of K,anchor=north,align=left] {K};
\node [above=of M1,anchor=north,align=left] {Dissolved \\modules};
\node [above=of M2,anchor=north,align=left] {Subset of\\ job sizes};
\node [above=of J,anchor=north,align=left] {Correspon- \\ding job};

\draw (K.text split east) -- (M1.north west);
\draw (K.two split east) -- (M1.south west);

\draw (M1.text split east) -- (M2.north west);
\draw (M1.south east) -- (M2.south west);

\draw (M2.text split east) -- (J.north west);
\draw (M2.two split east) -- (J.south west);

\end{tikzpicture}
\caption{This figure partially dissolves a configuration $K$ into a job $J_2^4$. First, a configuration consists of occurrences of module sizes $K_1, \dots, K_4$. For example the second module size appears $K_2 = 2$ times in the configuration $K$. Meaning we have a placeholder for two modules of exactly this size, here $L$ and $M$. Modules themselves hold multiplicities of job sizes, for example module $M$ has an occurrence of $M_2 = 1$ of the fourth job size. This space is then filled with the job $J_2^4$ of class 4 with precisely that size.} \label{fig:unfoldingConfig}
\end{center}
\end{figure}
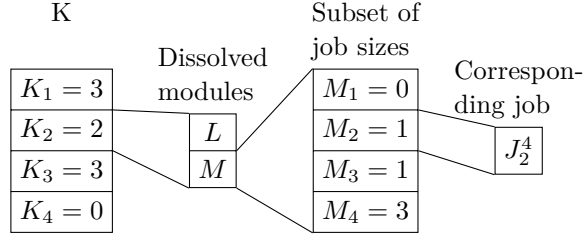

The overall running time for placing the large classes onto the machines is linear in the number of involved jobs, i.e. dissolving all configurations takes time $\cO(m \cdot 1/\delta)$. 
When placing the small jobs, we touch each small job and thus at most each class at most once. 
To insert the original jobs and job sizes, we have to consider each job once. 
As $m \leq n$ and $C \leq n$, this step yields an overall running time of $\cO(m \cdot 1/\delta + C + n) = \cO(n/\delta)$.

\subparagraph{Total running time and error.}

We have seen how we can also solve the non-preemptive case of the Class Constraint Scheduling problem when we are given a guess $T$ on the makespan. 
Indeed, a complete algorithm again requires to embed the algorithm given above in a binary search. 
As $n \cdot p_{\max}$ is an upper bound on the largest possible makespan and the optimal makespan is integral, the search is exhausted after at most $\cO(\log((n \cdot p_{\max}/\delta))$ steps. 
Again we can reasonably assume that $C \leq n$, we get a total running time of: 
\begin{align*}
\underbrace{\cO(\log(n \cdot p_{\max}/\delta))}_\text{Binary Search} \cdot [\underbrace{\cO(n)}_\text{Preprocess} + \underbrace{C^{\cO(1/\delta^8\log(1/\delta))} \cdot \log m}_\text{\nfold{}} + \underbrace{\cO(n/\delta^2)}_\text{Constructing the packing}] \\
= n^{\cO(1/\delta^8 \log(1/\delta))} \log(m) \log(p_{\max})
\end{align*}

Since we bounded the error with $\cO(\delta)$ in each step, the overall error is also at most $\cO(\delta)$. Setting $\epsilon = \cO(\delta)$ we get the desired approximation ratio. This completes the algorithm and its analysis. Overall, we get:

\begin{theorem}
A schedule $\sigma$ for the non-preemptive version of the Class Constrained Scheduling problem is obtained in time $n^{\cO(1/\delta^8 \log(1/\delta))} \log(m) \log(p_{\max})$ with makespan $\mu(\sigma_P) \leq (1+ \epsilon) \cdot \mu(\textsc{opt}(I))$, where OPT$(I)$ denotes a solution with optimal makespan for packing the instance $I$. This yields the desired PTAS $P$ for this problem.
\end{theorem}


\subsection{Preemptive Case}
In the preemptive case we are allowed to split jobs arbitrary as long as pieces belonging to the same job are not executed in parallel. This additional constraint makes it the hardest case. However, we handle these obstacles by proving some nice structure about an optimal solution. Using it we can then formulate the \nfold{}. Again, adapting the remaining steps of the algorithms presented before, we get the desired PTAS.

\subparagraph{Preprocessing.}

We perform the same preprocessing we did for the non-preemptive case and derive an instance $I'$ with a set of jobs $J'$ and the property that each class is either large or small.
We have:
\begin{lemma}\label{lem:preprocess_preemptive}
If there is a schedule with makespan $T$ for $I$, then there is also a schedule with makepan $(1+3\delta)T$ for instance $I'$ in which each job belonging to a small class is completely scheduled on one machine.
\end{lemma}
\begin{proof}
The proof is very similar to the proof of Lemma \ref{lem:preprocess_non-preemptive}.
The main difference is that we define $y_{i,j} \in [0,1]$ to be the fraction of job $j\in J$ that is scheduled on machine $i$.
Furthermore, when constructing the schedule for $I'$ from the $\bar{x}$-variables, we have to place the jobs from $J'\setminus J$ into the gaps left when removing the jobs from $J\setminus J'$.
Note that we do not have to change the approach to guarantee that each job belonging to a small classes is completely scheduled on one machine afterwards.
\end{proof}
The existence of a schedule in which each job belonging to a small class is completely scheduled on one machine is an important detail in the following. 

Furthermore, we round and scale the processing times and the makespan like in the two other cases and call the resulting instance $I''$.
For small classes $u$, we write $p''_u$ to denote the processing time of the single job of the class.
Like before, we denote the set of rounded processing times occurring in large classes by $\mathcal{P}$ and the number of jobs of class $u$ and size $p\in\mathcal{P}$ by $n_p^u$.

\subparagraph{Well-Structured Schedule.}

In the preemptive case, we call a schedule well-structured if the following two conditions hold:
\begin{itemize}
\item Each job belonging to a small class is completely scheduled on one machine.
\item For jobs belonging to large classes, each job piece starts at a multiple of $\delta^2 T$ and its size is a multiple of $\delta^2 T$. 
\end{itemize}
Given some fixed makespan bound $T'$, we define the set of layers $L$ as $\{ \ell\in\ints_{> 0} |  (\ell - 1)\delta^2 T \leq T'\}$. 
If a job $j$ is at least partially scheduled in the time window $[(\ell-1)\delta^2T, \ell\delta^2 T)$ for some $\ell\in\ints_{> 0}$, we say that job $j$ is placed in layer $\ell$.
Moreover, we call pairs of layers and machines slots, denote the set of slots as $S = M \times L$, and say that a job $j$ is placed in a slot $s = (i,\ell)$ if it is (partially) scheduled on $i$ in the layer $\ell$.
Obviously, in a well-structured schedule, pieces of jobs belonging to large classes that are placed in some slot have to fill the slot completely.

\begin{lemma}
If there is a schedule with makespan $T'$ for instance $I'$ in which each job belonging to a small class is completely scheduled on one machine, then there is also a well-structured schedule for $I'$ with a makespan of at most $T' + \delta^2T$.
\end{lemma}
\begin{proof}
Let there be a schedule with makespan $T'$ for instance $I'$ in which each job belonging to a small class is completely scheduled on one machine.
For each machine $i$, let $D_i$ denote the overall processing time of jobs belonging to large classes that is executed on $i$. 
Moreover, let $\hat{J}$ denote the set of jobs belonging to large classes, and for each job $j$ and machine $i$, let $\chi_{i,j} = 1$ if a job belonging to the same class as $j$ is scheduled on $i$ and $\chi_{i,j} = 0$ otherwise.
We construct a flow network with the following nodes:
\begin{itemize}
\item A source $\alpha$ and a sink $\omega$,
\item $x_j$ for each job $j\in \hat{J}$, 
\item $u_{j \times \ell}$ for each job $j\in \hat{J}$ and layer $\ell\in L$, 
\item $v_{i\times \ell}$ for each slot $(i,\ell) \in S$, 
\item and $y_i$ for each machine $i\in M$.
\end{itemize}
Furthermore, we have the following edges and capacities:
\begin{itemize}
\item $(\alpha, x_j)$ for each $j\in \hat{J}$ with capacity $p_j/(\delta^2 T)$ (Note that $p_j$ is a multiplicity of $\delta^2T$ due to the rounding and thus $p_j/(\delta^2 T)$ is integral),
\item $(x_j, u_{j \times \ell})$ for each $j\in \hat{J}$ and $\ell\in L$ with capacity $1$,
\item $(u_{j \times \ell}, v_{i\times \ell})$ for each $j\in \hat{J}$, $i\in M$ and $\ell\in L$ with capacity $\chi_{i,j}$,
\item $(v_{i\times \ell}, y_i)$ for each $i\in M$ and $\ell\in L$ with capacity $1$,
\item and $(y_i, \omega)$ for each $i\in M$ with capacity $\lceil D_i / (\delta^2 T) \rceil$. 
\end{itemize} 
The construction is summarized in Figure \ref{fig:flow_well-structured_preemptive}.
\begin{figure}
\centering
\scalebox{1.0}{
\begin{tikzpicture}
\pgfmathsetmacro{\lgap}{2.2} 
\pgfmathsetmacro{\vlgap}{2.6} 
\pgfmathsetmacro{\vbgap}{1.2} 
\pgfmathsetmacro{\vsgap}{0.7} 
\pgfmathsetmacro{\recgap}{0.39} 


\node[draw, label = {[label distance= 0 pt]180:{$\alpha$}}, circle, fill=black, inner sep=0pt, minimum width=4pt] (source) at (0:0) {};

\begin{scope}[xshift = \lgap cm]
\node[draw, circle, fill=black, inner sep=0pt, minimum width=4pt] (x1) at (0,\vbgap) {};
\node[draw, circle, fill=black, inner sep=0pt, minimum width=4pt] (xj) at (0:0) {};
\node[draw, label = {[label distance=0pt]180:{}}, circle, fill=black, inner sep=0pt, minimum width=4pt] (xn) at (0,-\vbgap) {};

\node at ($(x1)!0.39!(xj)$) {$\vdots$};
\node at ($(xj)!0.39!(xn)$) {$\vdots$};

\draw[dashed] ($(xn) + (-\recgap,-\recgap)$) rectangle ($(x1) + (\recgap,\recgap)$);

\draw[->, >={Latex[length=1.8mm]}] (source) to (x1);
\draw[->, >={Latex[length=1.8mm]}] (source) to node [midway, above, yshift = -2pt] {\small$\frac{p_j}{\delta^2 T}$} (xj) ;
\draw[->, >={Latex[length=1.8mm]}] (source) to (xn);

\end{scope}

\begin{scope}[xshift = 2*\lgap cm]

\begin{scope}[yshift = \vlgap cm]
\node[draw, circle, fill=black, inner sep=0pt, minimum width=4pt] (u11) at (0,\vsgap) {};
\node[draw, circle, fill=black, inner sep=0pt, minimum width=4pt] (u1l) at (0:0) {};
\node[draw, circle, fill=black, inner sep=0pt, minimum width=4pt] (u1L) at (0,-\vsgap) {};

\draw[dashed] ($(u1L) + (-\recgap,-\recgap)$) rectangle ($(u11) + (\recgap,\recgap)$);

\node at ($(u11)!0.39!(u1l)$) {$\vdots$};
\node at ($(u1l)!0.39!(u1L)$) {$\vdots$};

\end{scope}

\begin{scope}[yshift = 0 cm]
\node[draw, circle, fill=black, inner sep=0pt, minimum width=4pt] (uj1) at (0,\vsgap) {};
\node[draw, circle, fill=black, inner sep=0pt, minimum width=4pt] (ujl) at (0:0) {};
\node[draw, circle, fill=black, inner sep=0pt, minimum width=4pt] (ujL) at (0,-\vsgap) {};

\draw[dashed] ($(ujL) + (-\recgap,-\recgap)$) rectangle ($(uj1) + (\recgap,\recgap)$);

\node at ($(uj1)!0.39!(ujl)$) {$\vdots$};
\node at ($(ujl)!0.39!(ujL)$) {$\vdots$};

\draw[->, >={Latex[length=1.8mm]}] (xj) to (uj1);
\draw[->, >={Latex[length=1.8mm]}] (xj) to node [midway, above, yshift = -2pt] {\small $1$} (ujl);
\draw[->, >={Latex[length=1.8mm]}] (xj) to (ujL);
\end{scope}

\begin{scope}[yshift = -\vlgap cm]
\node[draw, circle, fill=black, inner sep=0pt, minimum width=4pt] (un1) at (0,\vsgap) {};
\node[draw, circle, fill=black, inner sep=0pt, minimum width=4pt] (unl) at (0:0) {};
\node[draw, circle, fill=black, inner sep=0pt, minimum width=4pt] (unL) at (0,-\vsgap) {};

\draw[dashed] ($(unL) + (-\recgap,-\recgap)$) rectangle ($(un1) + (\recgap,\recgap)$);

\node at ($(un1)!0.39!(unl)$) {$\vdots$};
\node at ($(unl)!0.39!(unL)$) {$\vdots$};

\end{scope}

\draw[dashed] ($(unL) + (-2*\recgap,-2*\recgap)$) rectangle ($(u11) + (2*\recgap,2*\recgap)$);
\end{scope}

\begin{scope}[xshift = 3.2*\lgap cm]

\begin{scope}[yshift = \vlgap cm]
\node[draw, circle, fill=black, inner sep=0pt, minimum width=4pt] (v11) at (0,\vsgap) {};
\node[draw, circle, fill=black, inner sep=0pt, minimum width=4pt] (v1l) at (0:0) {};
\node[draw, circle, fill=black, inner sep=0pt, minimum width=4pt] (v1L) at (0,-\vsgap) {};

\draw[dashed] ($(v1L) + (-\recgap,-\recgap)$) rectangle ($(v11) + (\recgap,\recgap)$);

\node at ($(v11)!0.39!(v1l)$) {$\vdots$};
\node at ($(v1l)!0.39!(v1L)$) {$\vdots$};
\end{scope}

\begin{scope}[yshift = 0 cm]
\node[draw, circle, fill=black, inner sep=0pt, minimum width=4pt] (vi1) at (0,\vsgap) {};
\node[draw, circle, fill=black, inner sep=0pt, minimum width=4pt] (vil) at (0:0) {};
\node[draw, circle, fill=black, inner sep=0pt, minimum width=4pt] (viL) at (0,-\vsgap) {};

\draw[dashed] ($(viL) + (-\recgap,-\recgap)$) rectangle ($(vi1) + (\recgap,\recgap)$);

\node at ($(vi1)!0.39!(vil)$) {$\vdots$};
\node at ($(vil)!0.39!(viL)$) {$\vdots$};
\end{scope}

\begin{scope}[yshift = -\vlgap cm]
\node[draw, circle, fill=black, inner sep=0pt, minimum width=4pt] (vm1) at (0,\vsgap) {};
\node[draw, circle, fill=black, inner sep=0pt, minimum width=4pt] (vml) at (0:0) {};
\node[draw, circle, fill=black, inner sep=0pt, minimum width=4pt] (vmL) at (0,-\vsgap) {};

\draw[dashed] ($(vmL) + (-\recgap,-\recgap)$) rectangle ($(vm1) + (\recgap,\recgap)$);

\node at ($(vm1)!0.39!(vml)$) {$\vdots$};
\node at ($(vml)!0.39!(vmL)$) {$\vdots$};
\end{scope}

\draw[->, >={Latex[length=1.8mm]}] (ujl) to (v1l);
\draw[->, >={Latex[length=1.8mm]}] (ujl) to node [midway, above, yshift = -2pt] {\small $\chi_{i,j}$} (vil);
\draw[->, >={Latex[length=1.8mm]}] (ujl) to (vml);

\draw[->, >={Latex[length=1.8mm]}] (u1l) to (vil);
\draw[->, >={Latex[length=1.8mm]}] (unl) to (vil);

\draw[dashed] ($(vmL) + (-2*\recgap,-2*\recgap)$) rectangle ($(v11) + (2*\recgap,2*\recgap)$);
\end{scope}

\begin{scope}[xshift = 4.2*\lgap cm]
\node[draw, circle, fill=black, inner sep=0pt, minimum width=4pt] (y1) at (0,\vbgap) {};
\node[draw, circle, fill=black, inner sep=0pt, minimum width=4pt] (yi) at (0:0) {};
\node[draw, circle, fill=black, inner sep=0pt, minimum width=4pt] (ym) at (0,-\vbgap) {};

\draw[dashed] ($(ym) + (-\recgap,-\recgap)$) rectangle ($(y1) + (\recgap,\recgap)$);

\node at ($(y1)!0.39!(yi)$) {$\vdots$};
\node at ($(yi)!0.39!(ym)$) {$\vdots$};

\draw[->, >={Latex[length=1.6mm]}] (vi1) to (yi);
\draw[->, >={Latex[length=1.6mm]}] (vil) to node [midway, above, yshift = -2pt] {\small $1$} (yi);
\draw[->, >={Latex[length=1.6mm]}] (viL) to (yi);
\end{scope}

\begin{scope}[xshift = 5.2*\lgap cm]
\node[draw, label = {[label distance=0 pt]0:{$\omega$}}, circle, fill=black, inner sep=0pt, minimum width=4pt] (sink) at (0:0) {};

\draw[->, >={Latex[length=1.6mm]}] (y1) to (sink);
\draw[->, >={Latex[length=1.6mm]}] (yi) to node [midway, above, yshift = -2pt, xshift = -4pt] {\footnotesize $\big\lceil\! \frac{D_i}{\delta^2\! T} \big\rceil$} (sink);
\draw[->, >={Latex[length=1.6mm]}] (ym) to (sink);
\end{scope}

\node at ($ (xj) + (0, \vbgap + 1.8*\recgap) $) {\small Jobs};
\node[circle, fill = white, minimum width=4pt, inner sep=0pt] (lxj) at ($ (xj) + (45:0.4) $) {$x_j$};

\node at ($ (ujl) + (0, \vsgap + \vlgap + 2.8*\recgap) $) {\small Jobs $\times$ Layers};
\node[circle, fill = white, minimum width=4pt, inner sep=0pt] (lxj) at ($ (ujl) + (38:0.45) $) {$u_{j,\ell}$};
\node[circle, fill = white, minimum width=4pt, inner sep=0pt] (lxj) at ($ (ujl) + (60:1.2) $) {$j$};

\node at ($ (vil) + (0, \vsgap + \vlgap + 2.8*\recgap) $) {\small Slots};
\node[circle, fill = white, minimum width=4pt, inner sep=0pt] (lxj) at ($ (vil) + (38:0.43) $) {$v_{i,\ell}$};
\node[circle, fill = white, minimum width=4pt, inner sep=0pt] (lxj) at ($ (vil) + (60:1.2) $) {$i$};

\node at ($ (yi) + (0, \vbgap + 1.8*\recgap) $) {\small Machines};
\node[circle, fill = white, minimum width=4pt, inner sep=0pt] (lxj) at ($ (yi) + (45:0.38) $) {$y_i$};

\end{tikzpicture}
}
\caption{The flow network used to proof the existence of a well-structured schedule. Only edges incident to the nodes on the middle vertical axes ($\alpha$, $x_j$, $u_{j,\ell}$, \dots) are added.}
\label{fig:flow_well-structured_preemptive}
\end{figure}
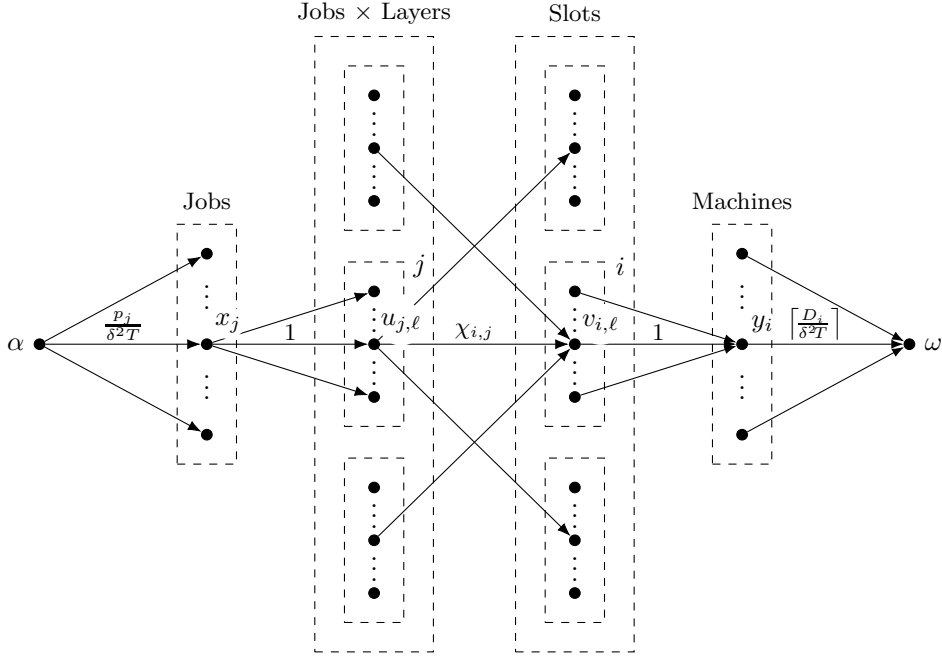
Note that all the capacities are integral and $\sum_{j\in \hat{J}} p_j/(\delta^2 T)$ is an obvious upper bound for a maximum flow in the network.
Let $p(i,j,\ell)$ be the processing time of job $j$ placed in slot $(i,\ell)$ in the given schedule.
It is not hard to verify that we get a feasible flow $f$ with value $\sum_{j\in \hat{J}} p_j/(\delta^2 T)$ by setting:
\begin{center}
\begin{tabular}{l|ccccc}
$\eta$ & $(\alpha, x_j)$ & $(x_j, u_{j \times \ell})$ & $(u_{j \times \ell}, v_{i\times \ell})$ & $(v_{i\times \ell}, y_i)$ & $(y_i, \omega)$\\\midrule
$f(\eta)$ &  $\frac{p_j}{\delta^2 T}$ & $\sum_{i\in M}\frac{p(i,j,\ell)}{\delta^2 T}$ & $\frac{p(i,j,\ell)}{\delta^2 T}$ & $\sum_{j\in \hat{J}}\frac{ p(i,j,\ell)}{\delta^2 T}$ & $\frac{D_i}{\delta^2 T}$\\
\end{tabular}
\end{center}
By flow integrality, there also exists an integral flow $\bar{f}$ with value $\sum_{j\in \hat{J}} p_j/(\delta^2 T)$.
We may use $\bar{f}$ in turn to define a schedule for the jobs of $\hat{J}$.
We do so by defining the processing time $\bar{p}(i,j,\ell)$ of job $j$ placed in slot $(i,\ell)$ in the new schedule.
Note that $\bar{f}(e)\in \set{0,1}$ for any edge $e$ of the second, third or fourth type.
We set $\bar{p}(i,j,\ell) = \bar{f}(u_{j \times \ell}, v_{i\times \ell}) \delta^2 T$.
Due to the structure of the flow network, we have:
\begin{itemize}
\item For each job $j$ and layer $\ell$ there is at most one machine $i$ with $\bar{p}(i,j,\ell) > 0$.
\item For each machine $i$ and layer $\ell$ there is at most one job $j$ with $\bar{p}(i,j,\ell) > 0$.
\item For each job $j$, we have $p_j/(\delta^2 T) = \sum_{i\in M}\sum_{\ell\in L} \bar{p}(i,j,\ell)$.
\item For each machine $i$, we have $ \sum_{j\in\hat{J}}\sum_{\ell\in L}\bar{p}(i,j,\ell) \leq D_i + \delta^2 T$
\end{itemize}
Hence, if we schedule the jobs belonging to small classes on the same machines as before and place them greedily into the gaps, we get a feasible schedule with makespan at most $T' + \delta^2 T$.
\end{proof}

\subparagraph{Setting up the \nfold{}.}

Taking the above steps into considerations, we set $\bar{T} = (1+3\delta)(1+\delta^2)T =(1+\cO(\delta))T$ and search for a well-structured schedule with makespan $\bar{T}$ via an \nfold{} IP.
We define the set of layers $L$ with respect to this makespan bound, that is, $L = \sett{\ell\in\ints_{> 0}}{(\ell - 1)\delta^2 T \leq \bar{T}}$.
In a well-structured schedule, jobs fill up whole slots on a given machine, and the slots filled up by jobs of a certain class may be distributed in any possible way on that machine.
Hence, we define modules in this context as 0-1-vectors indexed by the layers that include at least one 1, i.e., $\mathcal{M} = \set{0,1}^L\setminus\set{(0,\dots,0)^\top}$.
Moreover, we define configurations as 0-1-vectors indexed by the modules such that at most $c$ modules are chosen, and no two modules occupying the same layer are chosen, that is, $\mathcal{K} = \sett[\big]{K\in\set{0,1}^{\mathcal{M}}}{\Vert K \Vert_1\leq c, \forall \ell\in L: \sum_{M\in\mathcal{M}}K_M M_\ell \leq 1}$.
The size of a configuration $\Lambda(K)$ is determined by the number of filled up slot, i.e., $\Lambda(K) = \delta^2 T\sum_{M\in\mathcal{M}} K_M\Vert M\Vert_1$.
Note that $\Vert K\Vert_1 \leq |L| = \cO(1/\delta^2)$ for each $K\in\mathcal{K}$.
Correspondingly, we set $c^* = \min\set{c, |L|}$ and we define $\mathcal{K}(h,b) = \sett{K\in\mathcal{K}}{\Lambda(K) = h, \Vert K\Vert_1 = b}$ for each $h\in\Lambda(\mathcal{K})$ and $c\in[c^*]$.

Let $u\in[C]$ be a class.
We have a variable $y^u_M\in\set{0,\dots, m}$ for each module $M\in \mathcal{M}$ indicating how often $M$ is chosen to cover the jobs of class $u$.
Moreover, we introduce a variable $x_K^u \in\set{0,\dots, m}$ for each configuration $K\in \mathcal{K}$. 
Like before, the duplication of the latter variables does not carry meaning and is only used to obtain the desired \nfold{} structure of the constraint matrix. 
Furthermore, we have binary variables $ z^u_{h,b}\in\set{0,1}$ for each $h\in\Lambda(\mathcal{K})$ and $b\in[c^*]$ which are used to decide whether the class is assigned to a machine on which $b$ job pieces with overall size $h$ and belonging to large classes are scheduled.
Lastly, we introduce variables $a^u_{p,\ell}\in\set{0,\dots,m}$ for each processing time $p\in\mathcal{P}$ and layer $\ell$.
These variables are used to determine how many slots in a given layer $\ell$ are filled by jobs with size $p$ and belonging to class $u$.
The \nfold{} has the following constraints:

\begin{align*}
\tag{0}
\label{c0_p}
&\sum\limits_{u=1}^C \sum\limits_{K \in \mathcal K} x^u_K = m\\
\tag{1}
\label{c1_p}
&\sum\limits_{u=1}^C \sum\limits_{K \in \mathcal K} K_M  x^u_K = \sum\limits_{u=1}^C  y^u_{M} & \forall M \in \mathcal M\\
\tag{2}
\label{c2_p}
&\sum\limits_{u=1}^{C} z^u_{h,b} + b  \sum\limits_{u=1}^C \sum\limits_{K \in \mathcal K(h,b)} x^u_K \leq c  \sum\limits_{u=1}^C \sum\limits_{K \in \mathcal K(h,b)} x^u_K &\forall h \in \Lambda(\mathcal{K}),b \in [c^*]\\
\tag{3}
\label{c3_p}
&\sum\limits_{u=1}^{C} p''_u  z^u_{h,b} + h  \sum\limits_{u=1}^C \sum\limits_{K \in \mathcal K(h,b)} x^u_K \leq \bar{T} \sum\limits_{u=1}^C \sum\limits_{K \in \mathcal K(h,b)} x^u_K &\forall h \in \Lambda(\mathcal{K}),b \in [c^*]\\
\tag{4}
\label{c4_p}
&\sum\limits_{\ell \in L} a_{p,\ell}^{u} = (1-\xi_u)\frac{p}{\delta^2T}n_p^u  &\forall u \in [C], p\in \mathcal{P} \\
\tag{5}
\label{c5_p}
&\sum\limits_{M\in\mathcal{M}} M_\ell y^u_M = \sum_{p\in\mathcal{P}}a_{p,\ell}^u &\forall u \in [C],\ell\in L\\
\tag{6}
\label{c6_p}
&\sum\limits_{h\in \Lambda(\mathcal{K})} \sum\limits_{b\in [c^*]} z^u_{h,b} = \xi_u & \forall u \in [C]
\end{align*}
Note that the \nfold{} is very similar two the ones presented so far. 
The main difference lies in the changed definitions of modules and configurations and in the Constraints (\ref{c4_p}) and (\ref{c5_p}).
Due to Constraint (\ref{c4_p}), it is guaranteed that a proper number of slots is reserved to place all the jobs of a certain class and size.
Furthermore, these slots are proberly covered by modules because of Constraint (\ref{c5_p}).
\begin{lemma}\label{lem:preemptive_nfold_correct}
If there is a schedule with makespan $\bar{T}$ for instance $I''$, then there is also a solution to the above \nfold{} IP.
\end{lemma}
\begin{proof}
Given a well-structured schedule, the slots occupied by job pieces belonging to a large class define the modules included in $\mathcal{M}$. The combination of modules appearing on one machine then derives a configuration included in $\mathcal{K}$. 
We set the $x^1$-variables accordingly, and set the $x^u$ variables for $u\neq 1$ to $0$.
Let $u\in[c^*]$. If $u$ is a large class it contains only large jobs. These jobs are split into pieces of size $\delta^2 T$ starting at multiplicities of $\delta^2T$ defining the used modules for that class. We set the $y^u$ variables accordingly and the $z^u$ variables to $0$. Further we can count layer-wise the number of slots being used for one large processing time from any class deriving the values for $a^u_{p, \ell}$ variables. If, on the other hand, $u$ is a small class, then the whole class is scheduled on the same machine $i$. Let $K$ be the configuration corresponding to $i$, $h = \Lambda(K)$, and $b = \norm{K}_1$.
We set $z^u_{h,b} = 1$ and $z^u_{h',b'} = 0$ for each $(h',b')\neq (h,b)$.
Furthermore, we set all the $y^u$ variables to $0$.
It is easy to verify that this solution is feasible.
\end{proof}
Hence, if the \nfold{} has no feasible solution, we can reject the makespan guess $T$.

\subparagraph{Solving the \nfold{}.}

We bound the parameters $r, s, t$, $\Delta$ and $L$ in the application of Theorem \ref{nfold}.
First note that we have $s = |\mathcal{P}| + |L| + 1 = \cO(1/\delta^2)$ many locally and $r = 1 + |\mathcal{M}| + 2|\Lambda(\mathcal{K})|c^* = 2^{\cO(1/\delta^2)}$ globally uniform constraints. 
Moreover, there are $t = |\mathcal{K}| + |\mathcal{M}| + 3|\Lambda(\mathcal{K})|c^* + |P||L| = 2^{\cO(1/\delta^4)}$ many variables for each block (including slack variables).
The largest number $\Delta$ is again upper bounded by $\cO(c/\delta^2)$ and the encoding length $L$ of the largest number in the input by $L=\cO(\log(c/\delta^2 + m))$.
Summing up, we get the following running time (using $C\geq \max\set{2,c}$):
\[(rs\Delta)^{\cO(r^2s + s^2)} L \cdot Nt \log^{\cO(1)}(Nt) \leq C^{2^{\cO(1/\delta^2)}} \cdot \log m\]

\subparagraph{Constructing the Schedule.}
Again, we still have to construct a packing using the solution of the \nfold{} LP. First, we assign the configurations chosen by the $x$-variables onto the machines, \ie create the corresponding slots of size $\delta^2T$ at the layers of the belonging modules. Next, we reserve the slots for the corresponding classes according to the $y^u_M$ variables. Finally, we fill the job pieces belonging to large classes accordingly to the $a^u_{p,\ell}$ variables greedily by proceeding as follows: We go trough the layers $\ell \in L$ in an arbitrary order. Fill $a^u_{p,\ell}$ many job pieces of  $a^u_{p,\ell}$ different jobs with processing time $p$, class $u$ and the most unassigned job pieces of size $\delta^2T$ onto the machines which have slots reserved for that class. It is easy to verify that we have sufficient many slots for placing the job in this manner due to the constraints of the \nfold{} ILP. Further, the next theorem proves that this approach will assign all large jobs without conflict, \ie no job pieces belonging to the same large job will be assigned to the same layer. 

\begin{theorem}
We can greedily assign jobs accordingly to the $a^u_{p,\ell}$ variables of class $u$ with processing time $p$, respecting the $y^u_M$ variables, such that job pieces of the same class are not executed in parallel.
\end{theorem}
\begin{proof}
Suppose the opposite. At some layer $\ell$ there are w.l.o.g. two slots but just one job $j$ of class $u$ with two job pieces. This would imply, that we placed the last job piece of another job of that class in some layer before while having two job pieces of $j$. This contradicts the procedure of the greedy algorithm. It remains to prove, that this also cannot happen while filling the first layer. Indeed this would imply, that we only have one job of that processing time $p$ and class $u$ and thus constraint~(\ref{c4_p}) would only allow one placeholder on each layer. Altogether, this proves the theorem. 
\end{proof}

Next, we assign the small jobs similar to before by using the round robin approach. For each $h\in\Lambda(\mathcal{K})$ and $b\in[c^*]$, we assign the jobs of the small classes $u$ with $z_{h,b}^u = 1$ onto the machines with configurations $K\in\mathcal{K}(h,b)$ via round robin.
Due to (\ref{c2_p}), all the jobs can be placed by this procedure.
Furthermore, due to Lemma \ref{t:RoundRobin} and (\ref{c3_p}), this yields a schedule with makespan at most $\bar{T} + \delta T$.
Lastly, we have to use the original running times and jobs, which can be done using a greedy approach.

The overall running time for placing the large classes is linear in the number of involved job pieces, that is $\cO(m/\delta^2)$.
When placing the small jobs, we touch each class at most once, \ie it takes time $\cO(C)$. Further, when we insert the original jobs and job sizes, we have to consider each job and job piece in the schedule once. This yields an overall running time of $\cO(n(m/\delta^2 + C))$.

\subparagraph{Total running time and error.}
We have seen how we can also solve the preemptive case of the Class Constraint Scheduling problem when we are given a guess $T$ on the makespan. Following the idea of the algorithms above, we have to complete the algorithm given above by embedding it into a binary search. Again, $n \cdot p_{\max}$ states an upper bound. Further using the fact, that the optimal makespan is integral, the binary search is exhausted after at most $\cO(\log((n \cdot p_{\max}/\delta))$ steps. Using $C \leq n$, we get a total running time of: 
\begin{align*}
\underbrace{\cO(\log(n \cdot p_{\max}/\delta))}_\text{Binary Search} \cdot [\underbrace{\cO(n)}_\text{Preprocess} + \underbrace{C^{2^{\cO(1/\delta^2)}} \cdot \log m}_\text{\nfold{}} + \underbrace{\cO(n(m/\delta^2 + C))]}_\text{Constructing the packing} \\
= n^{2^{\cO(1/\delta^2)}} \log(m) \log(p_{\max})
\end{align*}

Setting $\epsilon = \cO(\delta)$, we get an overall error of $\epsilon$ as the error of each step is bounded by $\cO(\delta)$ as we argued above. This yields the complete algorithm and its analysis. Summarizing, we get:

\begin{theorem}
A schedule $\sigma$ for the preemptive version of the Class Constrained Scheduling problem is obtained in time $n^{2^{\cO(1/\delta^2)}} \log(m) \log(p_{\max})$ with makespan $\mu(\sigma_P) \leq (1+ \epsilon) \cdot \mu(\textsc{opt}(I))$, where OPT$(I)$ denotes a solution with optimal makespan for packing the instance $I$. This yields the desired PTAS $P$ for this problem.
\end{theorem}


\section{Open Questions} \label{CSS:OpenQuestions}
This paper managed to fill the gap of approximation algorithms for the Class Constrainted Scheduling problem. It introduced efficient approximation algorithms with a constant quality for each case. Further it also presented the first PTASs making it possible to solve the problem near-optimal in reasonable time. However, there are still some questions unsolved. First of all one could try to improve the qualities of the constant approximation algorithms. Furthermore, it is not excluded, that there also exists EPTASs for each of the cases. This would imply for our approach, that it is possible to formulate an \nfold{} without making the parameter $c$ appear in the constraint matrix. At this point, it seems hard to handle this obstacle as it is necessary to guarantee that we do not allot too many classes onto a machine. However, different approaches might yield the desired running time. 
Furthermore, if each class only contains one job, an EPTAS is known \cite{JansenLZ16} for the non-preemptive variant.
This result even holds if the number of class slots is dependent on the machines, that is, for each machine $i$ a number of class slots $c_i$ is part of the input.
Hence, it would be interesting to study the corresponding variants of CCS.
\bibliography{ref}

\end{document}